\documentclass[12pt]{elsarticle}
\usepackage[left=25mm, top=20mm, right=20mm, bottom=20mm]{geometry}
\journal{European Journal of Control}
\usepackage{newtxtext}
\usepackage{color}
\definecolor{darkgreen}{rgb}{0,0.4,0}
\usepackage[latin1]{inputenc}
\usepackage{amsmath}
\usepackage{graphicx}
\usepackage{amssymb}
\usepackage{verbatim}
\usepackage{epsfig}
\usepackage[labelfont=bf]{caption}
\usepackage{enumerate}
\usepackage{subfigure}
\usepackage{epstopdf}
\usepackage{mathrsfs}
\usepackage{float}
\usepackage{lineno,hyperref}
\usepackage{lipsum}
\usepackage{stfloats}
\usepackage{titletoc}

\newtheorem{assumption}{Assumption}
\newtheorem{lemma}{Lemma}
\newtheorem{remark}{Remark}

\newtheorem{proof}{Proof}
\newtheorem{proposition}{Proposition}
\newtheorem{fact}{Fact}

\def\begequarr{\begin{eqnarray}}
\def\endequarr{\end{eqnarray}}
\def\begequarrs{\begin{eqnarray*}}
\def\endequarrs{\end{eqnarray*}}
\def\begarr{\begin{array}}
\def\endarr{\end{array}}
\def\begequ{\begin{equation}}
\def\endequ{\end{equation}}
\def\lab{\label}
\def\begdes{\begin{description}}
\def\enddes{\end{description}}
\def\begenu{\begin{enumerate}}
\def\begite{\begin{itemize}}
\def\endite{\end{itemize}}
\def\endenu{\end{enumerate}}

\def\lef[{\left[\begin{array}}
\def\rig]{\end{array}\right]}
\def\qed{\hfill$\Box \Box \Box$}
\def\begcen{\begin{center}}
\def\endcen{\end{center}}
\def\begrem{\begin{remark}\rm}
\def\endrem{\end{remark}}

\def\beeq#1{\begin{equation}{#1}\end{equation}}
\def\begmat#1{\begin{bmatrix}#1\end{bmatrix}}
\def\begali#1{\begin{align}{#1}\end{align}}
\def\begalis#1{\begin{align*}{#1}\end{align*}}

\def\calc{{\mathcal C}}
\def\calf{{\mathcal F}}

\def\calb{{\mathcal B}}

\def\call{{\mathcal L}}
\def\cald{{\mathcal D}}
\def\calz{{\mathcal Z}}
\def\hal{\frac{1}{2}}

\def\sb{s_{\mathtt{b}}}
\def\calh{{\mathcal{H}}}
\def\bp{{s}}
\def\L2e{{\mathcal L}_{2e}}

\def\rea{\mathbb{R}}

\def\diag{\mbox{diag}}
\def\adj{\mbox{adj}}
\def\et{\epsilon_t}

\def\l2{{\mathcal L}_2}
\def\l2e{{\cal L}_{2e}}

\def\rea{\mathbb{R}}

\def\diag{\mbox{diag}}
\def\col{\mbox{col}}

\def\hgrad{\mathcal{G}_{\texttt{grad}}}

\def\TAC{{\it IEEE Trans. Automatic Control}}
\def\TIE{{\it IEEE Trans. Industrial Electronics}}

\def\IJC{{\it International Journal of Control}}
\def\CDC{{\it IEEE Conference on Decision and Control}}
\def\SCL{{\it Systems \& Control Letters}}

\def\CST{{\it IEEE Trans. Control Systems Technology}}

\usepackage{color}
\newcommand{\bow}[1]{{\color{blue} #1}}
\newcommand{\blue}[1]{{\color{blue} #1}}

\usepackage[prependcaption,colorinlistoftodos]{todonotes}

\begin{document}

\begin{frontmatter}

\title{\Large On Generation of Virtual Outputs via Signal Injection: Application to Observer Design for Electromechanical Systems}

\author[SJTU,Supelec]{Bowen Yi}\ead{b.yi@outlook.com}
\author[Supelec,ITMO]{Romeo Ortega}\ead{ortega@lss.supelec.fr}
\author[Supelec]{Houria Siguerdidjane}\ead{houria.siguerdidjane@centralesupelec.fr}
\author[Supelec]{Juan E. Machado}\ead{juan.machado@l2s.centralesupelec.fr}
\author[SJTU,PC]{Weidong Zhang*}\ead{wdzhang@sjtu.edu.cn}

\address[SJTU]{Department of Automation, Shanghai Jiao Tong University, Shanghai 200240, China}  %
\address[Supelec]{Laboratoire des Signaux et Syst\`emes, CNRS-CentraleSup\'elec, Gif-sur-Yvette 91192, France}
\address[ITMO]{Department of Control Systems and Informatics, ITMO University, St. Petersburg 197101, Russia}
\address[PC]{Peng Cheng Laboratory, Shenzhen 518066, China}

\begin{keyword}
Signal Injection; Observer Design; Nonlinear Systems; Electromechanical Systems.
\end{keyword}
%

\begin{abstract}
Probing signal injection is a well-established technique to extract additional information from a weakly (or non) observable dynamical system. Using averaging theory, a framework to analyse such schemes  for general nonlinear systems has been recently proposed in [Combes {\em et. al.}, 2016], where it is shown that the signal injection may be used to generate a new high frequency component of the systems output that can be used for state observation or controller design. A key step for the success of this technique is the implementation of a filter to reconstruct this virtual output from the measurement of the overall systems output. The main contribution of this paper is to propose a new filter with guaranteed convergence properties that outperforms the classical designs. The method is applied to a general class of electromechanical systems, and its performance is assessed via simulations and experiments on the benchmark example of a 1-dof magnetic levitation system.
\end{abstract}

\end{frontmatter}
%


%
\section{Introduction}
\label{sec1}
%

Many high-performance controller design techniques for nonlinear systems rely on the availability of the systems state. In many practical applications, the installation of sensors is stymied by cost and technological considerations. Observer design can then be invoked to reconstruct the state via input and output measurements. It is clear, however, that the system must satisfy some observability property for the success of this approach, see {\em e.g.}, \cite{BESbook,GAUKUP}. In the case when the latter is not satisfied, it is still possible to extract additional information from the system via probing signal injection---that is a well-established technique widely used in several applications.

A typical example is the estimation of position in electrical motors, which are non-observable at zero speed \cite{GLUDEL,MARbook}.  It is well-known that injecting a high frequency signal in the voltage and measuring the output current it is possible to generate a reasonable estimate of the rotor angle \cite{HOLtie,NAM,VAS}. Indeed, due to existence of rotor saliency and/or flux saturation, the injected signal at the voltage induces a decomposition of the measured current into a low and a high frequency component, from knowledge of the latter it is possible to recover the rotor angle, which is a necessary information for high performance controller design. Clearly, a key step for the successful application of this technique is the identification of the high frequency component of the current, a task that is typically accomplished with a combination of linear time-invariant (LTI) high-pass  filters and low-pass filters \cite{NAM}.

Another example of practical interest is magnetic levitation (MagLev) systems, where position control of the levitated object is of paramount importance and existing position sensors are expensive and unreliable. The interested reader is referred to \cite{GLUetalcep,MASetalifac,MIZetalcst} for a review of the existing literature on sensorless control of MagLev systems reported in the control community and to \cite{RANetaltia,SCHMAS} for results found in applications literature.

The main contribution of this paper is to propose a new filtering technique to identify the high frequency component of the output induced by the signal injection, which is applicable for general nonlinear systems. Instrumental for our developments is the use of the mathematical formalism proposed in the recent interesting paper \cite{COMetalacc} where, invoking second order averaging theory \cite{SANbook}, it is shown that injecting a periodic high-frequency signal in the systems input generates an output consisting of the sum of a low and a high frequency component---the first one corresponding to the output of the systems average dynamics and the second one called virtual output. A sliding-window filtering technique to identify the virtual output is then proposed in  \cite{COMetalacc}, which is used to design an (augmented) output-feedback control law. See also \cite{YIetalscl} where a slight extension of this technique is used to enlarge the domain of applicability of the parameter estimation-based observer proposed in \cite{ORTetalscl}.

The filter design technique proposed in this paper relies on the following two key observations. First, that the task of reconstructing the virtual output can be recast as a problem of estimation of parameters in a linear regression model.\footnote{See \cite{FORetal} where a similar ``identification-based" approach is pursued within the context of robust output regulation.} Second, the observation that the particular form of the regressor can be exploited to apply the  dynamic regressor extension and mixing (DREM) estimator proposed in \cite{ARAetaltac}, with some suitable operators that, on one hand, generate extended regressor and, on the other hand, guarantee the excitation conditions needed for exponential parameter convergence. Three significant advantages of the proposed filter are, first, that the exponential stability property makes the filter robust, a property that should be contrasted with the sensitivity to measurement noise observed for the sliding-window filter of \cite{COMetalacc}. Second, that due to the fact that the filter implementation relies on the use of linear time-varying (LTV) filters, it has a very simple practical implementation. Third, the new filter has a clear connection with the standard approach of high-pass/low-pass filtering universally adopted in practice, simplifying in this way the communication with the applied community.

In the second part of this paper we show how the new filtering technique can be applied to estimate the electrical coordinates of a general class of electro-mechanical systems (EMS), assuming that only the current and the voltage are measurable. This observation step is essential for the design of sensorless (also called self-sensing) controllers, which is a topic of great interest to the applied \cite{CHOetaltpe,HOLtie,RANetaltia,SCHMAS} and the control theory \cite{BERPRA,COMetalacc,GLUetalcep,MARbook,MASetalifac,MIZetalcst,ORTetalcst,VERetal} communities. Indeed, it is widely recognized that in motors, as well as MagLev systems, the key step to observe the mechanical coordinates is the reconstruction of the electrical coordinates, {\em i.e.}, fluxes and charges. The new observer is applied to the optical switch and the one-degree-of-freedom (1-dof) MagLev systems, with illustrative experiments carried out for the latter.

The remainder of the paper is organized as follows. Section \ref{sec2} gives some preliminaries on the analysis techniques of \cite{COMetalacc} and the DREM estimator of  \cite{ARAetaltac}. The main result of the paper, namely a new filter to reconstruct the virtual output, is presented in Section \ref{sec3}. In Section \ref{sec4} we apply this filter to a class of EMS, which includes the practical examples discussed in Section \ref{sec5}, where we present a detailed discussion, including experimental evidence, of its application for the critically important case of a 1-dof MagLev system. The paper is wrapped-up with concluding remarks and future research directions in  Section \ref{sec6}.
\\ \ \\
\textbf{Caveat.} An abridged version of this paper was reported in \cite{YIetalcdc}.
\\ \ \\
\textbf{Notation.}  $\epsilon_t$ is an exponentially decaying term with a proper dimension. $I_q$ is the $q$-dimensional identity matrix. $\det\{A\}$ and $\adj\{A\}$ represent the determinant and the adjunct matrix of a square matrix $A$. The Laplace transform symbol $s$ is used also to denote the derivative operator ${d \over dt}$. For an operator $\mathcal{H}$ acting on a signal we use the notation $\calh[\cdot](t)$, when clear from the context, the argument $t$ is omitted. $\mathcal{O}$ is the uniform big O symbol, that is, $f(z,\varepsilon)$ if and only if $|f(z,\varepsilon)|\le C\varepsilon$ for a constant $C$ independent of $z$ and $\varepsilon$. All mappings are assumed smooth enough. We define the operator $\nabla :=(\partial  /\partial x)^\top$.


%
\section{Preliminaries}
\lab{sec2}
%
In this section we briefly review the two main tools used for the development of  the proposed virtual output filter.
\subsection{Signal injection and virtual outputs}
\label{subsec21}
Let us first recall some results on the signal injection method proposed in \cite{COMetalacc}. Consider the nonlinear system
\begin{equation}
\label{sys}
\dot{x}  = f(x) + g(x)u,\quad
y  = h(x),
\end{equation}
where $x\in \rea^n,\;y\in \rea^{m}$, and $u \in\rea^{p}$. To generate the virtual output we apply the following input to the system
\begin{equation}
\label{eq1}
 u  = u_C + \sb, \quad \sb = s\bigg( {t \over \varepsilon}\bigg) \frak{b},
\end{equation}
where $u_C$ is the plants nominal input, typically the output of the {controller}, and $\sb$ is a high-frequency signal, with $s(\cdot)$ a 1-periodic, zero mean function, $\varepsilon \in (0,1)$ is a small constant, and $\frak{b} \in\rea^{p}$ is a free constant ``scaling'' vector. A key result of the signal injection method in \cite{COMetalacc}, which is established by second-order averaging analysis, is as follows. See also \cite{YIetalscl} for the multi-input case with application to parameter estimation-based observer design.

\begin{proposition}
\rm \label{prop1}
Consider the system \eqref{sys}, \eqref{eq1} where $u_C$ is a signal such that all state trajectories are bounded.  There exists $\varepsilon^*>0$, such that $\forall \varepsilon\in(0,\varepsilon^*]$,\footnote{When the closed-loop system under the feedback $u=u_C(x)$ is asymptotically stable, \eqref{higlow} is true in the time interval $[0,\infty)$; in the general case, it is true in $[0, \mathcal{O}({1 \over \varepsilon}))$.}
 \begali{
 \lab{higlow}
 x = \bar{x} + \varepsilon S  g(\bar x)\frak{b} + \mathcal{O}(\varepsilon^2),
 }
 is satisfied, where
 \begequ
 \lab{S}
 S(t) := S_0\bigg({t\over \varepsilon}\bigg),\quad
 S_0(t)  := \int_{0}^{t} s(\tau) d\tau - \int_{0}^{1} \int_{0}^{\sigma}s(\tau)d\tau d\sigma,
 \endequ
and $\bar x$ is generated as
$
\dot {\bar x}=f(\bar x)+g({\bar x})u_C
$
with $\bar x(0) = x(0)$, assuming $\int_{0}^{1} \int_{0}^{\sigma}s(\tau) d\tau d\sigma=0$. Furthermore, we have the identity
\begequ
\label{rel1}
y = \bar{y} + \varepsilon S y_v + \mathcal{O}(\varepsilon^2),
\endequ
where
\begequ
\label{yv}
\bar y :=h({\bar x}), \quad y_v  := \nabla h^\top(x)  g(x) \frak{b}.
\endequ
\hfill{$\triangleleft$}
\end{proposition}
%

Notice that if we substitute $t=0$ into \eqref{higlow}-\eqref{S}, it is easy to see that
$
\int_{0}^{1} \int_{0}^{\sigma}s(\tau) d\tau d\sigma=0
$
is a necessary condition for $x(0) = \bar{x}(0)$. Hence, the last term in the primitive function $S_0(\cdot)$ disappears.

The next step in the signal injection method is to estimate, from the measurement of $y$, the virtual output $y_v$. To simplify the notation, and with some obvious abuse of notation, in the sequel we omit the clarification that the averaging analysis only ensures the existence of an upper bound on $\varepsilon$ such that \eqref{rel1} holds, and we simply assume that $\varepsilon$ is small enough.
\subsection{Dynamic regressor extension and mixing}
\label{sec2-2}
DREM is a novel approach for estimation of the parameters in a regression model proposed in \cite{ARAetaltac}. The main feature of DREM is that it allows to generate $q$ {\em scalar} regression models, where $q$ is the dimension of the unknown parameter vector. In this way, parameter convergence is guaranteed without the standard persistency of excitation (PE) assumption\footnote{We recall that a bounded vector signal $\phi \in \rea^q$ is said to be PE if there exist $\delta>0$ and $T>0$ such that $\int_t^{t+T}\phi(\tau)\phi^\top(\tau)d\tau \geq \delta I_q$ for all $t \geq 0$.}  on the regressor, which is necessary in classical gradient or least-squares estimators \cite{SASBOD}.

The main result of DREM for linear regressions is summarized in the following proposition.

\begin{proposition}
\rm\label{prop_DREM} \cite{ARAetaltac}
Consider the $q$--dimensional linear regression
\begequ
\label{y}
c=\phi^\top \theta,
\endequ
where $c \in \rea$ and $\phi \in \rea^q$ are known, bounded functions of time and $\theta \in \rea^q$ is a vector of unknown, constant parameters. Introduce a {linear, single-input $q$-output, $\call_\infty$--stable} operator $\calh: \call_\infty \to \call^q_\infty,$ and define the signals $C \in \rea^q$ and $\Phi \in \rea^{q \times q}$ as
$
C(t)  := \calh [c](t), ~
{\Phi(t)  :=\big[ \calh [\phi_1](t) ~|~ \ldots ~|~ \calh [\phi_q](t) \big]  } .
$
The gradient-descent estimator\footnote{For brevity, the clarification $i \in \{1,2,\dots,q\}$ is omitted in the sequel.}
$$
\dot{\hat{\theta}}_i = \gamma_i\Delta (\mathcal{C}_i  - \Delta \hat{\theta}_i),\;i \in \{1,2,\dots,q\},
$$
with adaptation gains $\gamma_i>0$, and the signals $\Delta \in \rea$ and $\calc \in \rea^q$ defined as $\Delta :=\det \{\Phi\}$ and $\mathcal{C} := \adj\{\Phi\} C$, guarantees
\begin{itemize}
\item (element-wise parametric error monotonicity) The estimation error $\tilde \theta:=\hat \theta - \theta$ satisfies
$$
|\tilde \theta_i(t_b)| \geq |\tilde \theta_i(t_a)|,\;\forall\; t_a \geq t_b \geq 0;
$$
\item (condition for parameter convergence)
$$
 \lim_{t\to \infty} \tilde \theta_i(t)=0\quad \Longleftrightarrow \quad \Delta \notin \call_2.
$$
\end{itemize}
Moreover, if $\Delta$ is PE then the parameter convergence is {\em exponential}.
\hfill{$\triangleleft$}
\end{proposition}

The elements of the operator $\mathcal{H}$ may be simple, exponentially stable {LTI filters} of the form
$
\calh_i({s})=\frac{\alpha_i}{{s} + \beta_i},
$
with $\alpha_i\ne 0$, $\beta_i>0$. Another option of interest is delay operators, that is
$
[\calh_i(\cdot)](t):=(\cdot)(t-d_i),
$
where $d_i>0$. See \cite{ORTetalauto} for the case of general LTV operators and the connection of DREM with Luenberger functional observers.

%
\section{A New Procedure to Reconstruct $y_v$}
\lab{sec3}
%
In this section, we give the main result of this note, namely, a DREM-based filter to reconstruct the virtual output $y_v$.

\subsection{A linear regressor viewpoint}
\lab{subsec31}
The first step to apply DREM is to obtain the linear regression model. For, we make the key observation that there exists a time-scale separation between the probing signal $S$ that, by definition, has a high frequency, and the signals $\bar y$ and $y_v$. This motivates us to view \eqref{rel1} as an LTV regression perturbed by a small term $\mathcal{O}(\varepsilon^2)$. Whence,  we write \eqref{rel1} as
\begequ
\label{iphithe}
\begin{aligned}
    y & = \phi \theta+ \mathcal{O}(\varepsilon^2)\\
    \theta & =\begmat{\theta_1 \\ \theta_2}:= \begmat{\bar{y} \\ \varepsilon y_v} \in \rea^{2m},\;
    \phi := \begmat{I_{m} & S(t)I_{m}} \in \rea^{m \times 2m},
\end{aligned}
\endequ
with $y$ the measurable signal, and $\phi$ and $\theta$ playing the roles of known regressor and (slowly time-varying) parameters to be estimated.

\subsection{Generation of a linear regressor for $\theta_2$ only}
\lab{subsec32}
Although it is possible to show that $\phi$ is PE and, consequently, apply a gradient estimator to the linear regression \eqref{iphithe}, transient performance can be improved observing that we are interested in reconstructing only $y_v$, {\em i.e.}, only the estimation of $\theta_2$ is of interest. To achieve this end, we need to construct a new linear regression where only the parameter $\theta_2$ appears, which is possible following the DREM methodology, with suitably chosen operators.

Similarly to \cite{COMetalacc,YIetalscl} we consider the use of the weighted zero-order hold (WZOH) operator $\calz_w$, which is parameterized by $w>0$ and, acting on an input signal $v$, is defined as
\begequ
\lab{zw}
\dot{\chi}(t)  = v(t), \quad
\calz_w[v](t)  ={1\over w} \big[ \chi(t) - \chi(t- w) \big].
\endequ
The lemma below describes the action of the WZOH operator on the signal \eqref{rel1}.

\begin{lemma}
\lab{lem1}\rm
{\cite{COMetalacc}} Consider the signal \eqref{rel1} and the operator $\calz_w$ given in \eqref{zw}. Then,
\begequ
\lab{cla}
\calz_w[y](t) = \bar{y}(t - {w\over 2}) + \mathcal{O}(\varepsilon^2)
\endequ
holds for any $w= n\varepsilon$ with positive integer $n$.
\hfill{$\triangleleft$}
\end{lemma}

In words, Lemma \ref{lem1} shows that with the WZOH operator we can extract from $y$ the signal $\bar y$, with a delay of ${\omega \over 2}$. Now, from \eqref{rel1} we see that the action of a delay operator
$\cald_d$, with parameter $d>0$,
\begequ
\lab{hd}
\mathcal{D}_d[v](t)= v(t-d)
\endequ
on $y$ yields
\begequ
\lab{ddony}
\cald_d[y](t) = \bar{y}(t-d) + \varepsilon S(t-d) y_v(t-d) + \mathcal{O}(\varepsilon^2).
\endequ

The desired linear regression for $\theta_2$ only is given in the following fact, whose proof is established from direct inspection of \eqref{cla} and \eqref{ddony}.

\begin{fact}\em
Consider \eqref{rel1} and define the signal
\begequ
\label{Y}
Y(t) := \mathcal{D}_d[y](t) -  \mathcal{Z}_{2d}[y](t),
\endequ
with the operators $\calz_{2d}$ and $\cald_d$ defined in \eqref{zw} (with $w=2d$) and \eqref{hd}, respectively. Then,
\begequ
\lab{scareg}
Y(t) = S(t-d) \theta_2(t-d) + \mathcal{O}(\varepsilon^2).
\endequ
\end{fact}
We make the important observation that the role of the regressor in \eqref{scareg} is played by the scalar signal $S$, which is ``rich" by construction. More precisely, it satisfies
\begequ
\lab{sispe}
    \int_{t}^{t+ {1 \over \varepsilon}}S^2(\tau) d\tau \ge \delta_0,
\endequ
for all $t \geq 0$ and some $\delta_0>0$.

It is interesting to note that the regressor \eqref{scareg} can also be obtained applying {\em verbatim} the DREM procedure outlined in Proposition \ref{prop_DREM} with the $\call_\infty$-stable operator
$
\mathcal{H}: = \col( \mathcal{D}_{{d}} , \mathcal{Z}_{2d}).
$
For the sake of clarity, we have opted to present the alternative derivation of  \eqref{scareg} given above.
\subsection{DREM-based virtual output estimator}
\label{subsec33}
%
We are in position to propose the main result of the paper, that is a DREM-based filter to identify the virtual output $y_v$ from \eqref{rel1}. To present the proposition we need the following.

\begin{assumption}
\lab{ass1} \em
There exists a constant $c_v$, independent of $\varepsilon$, such that
$
|\dot{y}_v| \le c_v,
$
where $y_v$ is given in \eqref{yv}.
\hfill{$\triangleleft$}
\end{assumption}
%
\begin{proposition}
\label{pro1}\rm
Consider the system \eqref{sys}-\eqref{eq1} with bounded state trajectories verifying Assumption \ref{ass1}. Define the virtual output estimator
\begequ
\label{vir_out_filter}
    \dot{\hat{\theta}}_2  = \gamma S(t) \big[Y(t) - S(t) \hat{\theta}_2\big] , \quad
        \hat{y}_v  = {1 \over \varepsilon } \hat{\theta}_2,
\endequ
where $\gamma> {\gamma_\star \over \varepsilon}$ is a tuning gain for some $\gamma_\star>0$, $S$ is given in \eqref{S}, $Y$ in \eqref{Y}, and $\mathcal{Z}_{2d}$, $\mathcal{D}_d$,  are defined in \eqref{zw} and \eqref{hd}, respectively, with $d=\varepsilon$. Then, for any $c_v$ there always exists $\varepsilon$ globally guaranteeing
\begequ
\lab{tilyv}
\lim_{t\to \infty} |\hat{y}_v(t) - y_v(t) | \le \mathcal{O}({\varepsilon}) \quad \text{(exp.)}.
\endequ
\hfill{$\triangleleft$}
\end{proposition}

\begin{proof}
\rm
Define the error signal
$$
\Tilde{\theta}_2:= \Hat{\theta}_2 - \varepsilon y_v .
$$
Notice that---because of the periodicity of $S$---with the choice $d=\varepsilon$, \eqref{scareg} is equivalent to
$$
Y(t) = S(t) \theta_2(t-d) + \mathcal{O}(\varepsilon^2).
$$
Replacing the latter in \eqref{vir_out_filter}, and invoking Lemma \ref{lem1} and Assumption \ref{ass1},  we get
\begequ
\label{tilde_theta2}
\dot{\tilde{\theta}}_2= - \gamma S(t)[S(t)(\hat{\theta}_2(t)-\theta_2(t-d)) + \mathcal{O}(\varepsilon^2)]  + \varepsilon \dot{y}_v.
\endequ
For the second term on the right-hand side of \eqref{tilde_theta2}, noting the tininess of $\varepsilon>0$ and using the Taylor expansion element-by-element we have
$$
\begin{aligned}
\theta_2(t-d) & = \theta_2(t) -  \dot{\theta}_2(t) d + \mathcal{O}(\varepsilon^2) \\
& = \theta_2(t) {- \dot{y}_v(t)} \varepsilon^2 + \mathcal{O}(\varepsilon^2).
\end{aligned}
$$
For any $c_v$ we can always find $\varepsilon \in (0,1)$ to guarantee $| \dot{y}_v(t)|\varepsilon^2 \le c_v \varepsilon^2= \mathcal{O}(\varepsilon^2)$.
The equation \eqref{tilde_theta2} becomes
\begequ
\label{dyn_theta}
\dot{\tilde{\theta}}_2= - \gamma S^2(t)\tilde{\theta}_2 + \gamma \Delta_1 + \Delta_2.
\endequ
with
$$
\begin{aligned}
\Delta_1& := -S(t)^2\dot{y}_v \varepsilon^2+ \big(S(t)^2 + S(t)\big) \mathcal{O}(\varepsilon^2) \\
\Delta_2& := \varepsilon \dot{y}_v
\end{aligned}
$$
satisfying $||\Delta_1||_\infty \le c_\ell \varepsilon^2$ for some $c_\ell >0$. We define a new time scale $\tau$ as ${d\tau \over d t} = \gamma$, in which the error dynamics \eqref{dyn_theta} becomes
\begequ
\label{error_tau}
{d \tilde{\theta}_2 \over d\tau} = - S^2 \Big({1 \over \gamma}\tau \Big) \tilde{\theta}_2 + \Delta_1 + {1\over \gamma} \Delta_2.
\endequ
We notice that the last term satisfies 
$$
\left\| {1\over \gamma} \Delta_2\right\|_\infty \le {1\over \gamma_\star} c_v\varepsilon^2
$$ 
due to $\gamma > {\gamma_\star \over \varepsilon}$ with $\gamma_\star$ independent of $\varepsilon$.

{Recalling \eqref{sispe}, the origin of the unperturbed part of \eqref{error_tau}, {\em i.e.,} ${d {\tilde{\theta}}_2 \over d\tau}= - S^2({\tau \over \gamma})\tilde{\theta}_2$, is exponentially stable.} Using the converse Lyapunov theorem {\cite[Theorem 4.14]{KHA}}, and carrying-out some basic perturbation analysis {\cite[Lemma 9.2]{KHA}}, we complete the proof.
\qed
\end{proof}
\subsection{Discussion}
\label{subsec32}
The following remarks are in order.

\noindent {\bf R1} The virtual output filters in \cite{COMetalacc,YIetalscl} compute estimates by averaging in a (finite) moving horizon $[t-n\varepsilon,t]$ the observation error. The new filter \eqref{vir_out_filter} provides an alternative closed-loop approach, {which is similar to defining a moving average in infinite-time interval. We present some comparisons among these designs by simulations in Section \ref{sec5}.}

\noindent {\bf R2}   Consider the output with measurement noise, that is,
\begequ
\label{y_eq}
y = \bar{y} + \varepsilon S y_v + \mathcal{O}(\varepsilon^2) + \xi,
\endequ
where $\xi$ represents high-frequency measurement noise. Increasing $\varepsilon$ or the norm of $\mathfrak{b}$ can increase the signal-to-noise ratio, but at the price of degrading estimation accuracy. {Similar remarks were made in \cite{COMetalacc,PETetal}, carrying out power spectral density analysis in the stochastic framework and} sensitivity analysis in frequency domain, respectively. In \cite{YIetaltie} this tradeoff has been observed in some experimental evidence on motors, which is also the case for the experiments on the 1-dof MagLev system presented in Section \ref{sec5}. Also, notice that we can re-write the estimator \eqref{vir_out_filter} as
$$
\dot{\hat{y}}_v = \gamma' S(t) [Y(t) - \varepsilon S(t) \hat{y}_v],
$$
where the \emph{tuning gain} is $\gamma':= {\gamma \over \varepsilon}$ without affecting the analysis, hence avoiding the division by the small parameter $\varepsilon$ in the computation of $\hat y_v$.

\noindent {\bf R3} In Proposition \ref{prop1}, we require that $\gamma >{1\over \varepsilon} \gamma_\star$ for some $\gamma_\star>0$ independent of $\varepsilon$. For such a case, the parameter $\gamma$ has very limited effects on the ultimate accuracy in the presence of measurement noises, but only assigns the convergence speed. To show this, we write the filtered signal, via the filter $(\cald_d - \calz_{2d})[\cdot]$, of $\xi$ in \eqref{y_eq} as $\xi_f$, yielding $Y(t) = S(t-d) \theta_2(t-d) + \xi_f + \mathcal{O}(\varepsilon^2)$. Due to the BIBO property of $\cald_d$ and $\calz_{2d}$, the vector $\xi_f$ is also bounded. Hence, the error dynamics \eqref{dyn_theta} reads
\begequ
\label{error_dyn_noise}
{d {\tilde{\theta}}_2 \over d\tau} = -  S^2\Big({1\over \gamma}\tau\Big)\tilde{\theta}_2 +  S\Big({1\over \gamma}\tau\Big) \xi_f +
\Delta_1 + {1\over \gamma} \Delta_2 .
\endequ
Following the proof in \cite[Proposition 1]{ORTYI}, we are able to construct a strict Lyapunov function $V(\cdot)$ for the dynamics $\dot{\tilde{\theta}}_2= -  S^2(\cdot)\tilde{\theta}_2$, and then calculate its derivative  in the $\tau$ time scale. It implies that $\tilde{\theta}_2$, equivalently $\varepsilon \hat{y}_v$, converges to the invariant set
$$
\Omega:=\Big\{ \textbf{x} \in \rea^m ~\Big|~ |\textbf{x}| \le  k_1 \|S(t) \xi_f\|_\infty + k_2 \varepsilon^2 + {1\over \gamma_\star} k_3 \varepsilon^2  \Big\}
$$
with $k_i>0~(i=1,2,3)$ independent of $\varepsilon$. It is clear that the accuracy, determined by the set $\Omega$, hardly changes with different $\gamma > {\gamma_\star \over \varepsilon}$.

\noindent {\bf R4} Using the Laplace transform, the relationship \eqref{Y} may be represented in the frequency domain as
$
Y(s)=G_d(s) y(s),
$
where we defined the transfer function
$$
 G_d(s) := e^{-ds} +\dfrac{1}{2ds}\left( e^{-2ds} -  1\right).
$$
It is shown in \cite{YIetalccta,YIetaltie} that, for small $d>0$, this transfer function is a high-pass filter---with respect to the frequency content of the signals of interest. This provides the connection between the proposed filter and the classical filtering techniques widely used in applications \cite{HOLtie,NAM}.

\noindent {\bf R5} A block diagram realization of the proposed filter is given in Fig. \ref{fig1}, where we defined the LTV operator
$
\hat y_v=\hgrad[Y](t)
$
to represent \eqref{vir_out_filter}. It is clear, then, that the computational burden required for the practical implementation of the proposed filter is negligible.

\begin{figure}[htp!]
    \centering
    \includegraphics[width=12cm]{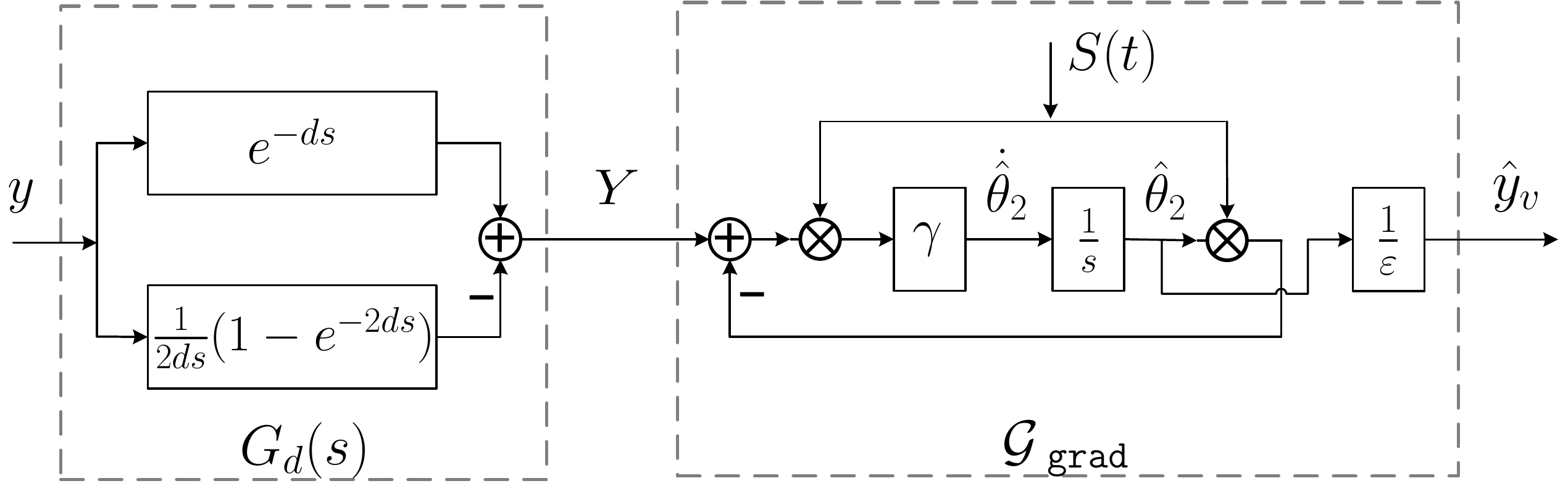}
    \caption{Block diagram of the proposed estimation method}
    \label{fig1}
\end{figure}

%
\section{An Observer of Fluxes and Charges in EMS}
\lab{sec4}
%
In this section we illustrate the application of the new filter to a general class of EMS and use the estimated virtual output to design an observer of its {\em fluxes and charges}, that is the states of the energy storing capacitive and inductive elements, which are usually not available for measurement. Equipped with the latter it is often possible to design an  observer for the mechanical coordinates, which are essential for the implementation of sensorless control, {in which we assume that only currents and voltages, are measurable by sensors.} See \cite{PYRetal} where flux observers, without signal injection, for EMS---with magnetic energy only---are proposed, and \cite{BOBetal} where full state observers, also without signal injection, for MagLev systems are developed.

\subsection{Model of the system and virtual outputs}
\lab{subsec41}
%
In this subsection we present the mathematical model of the EMS considered in the section. For the sake of brevity, the presentation is quite succinct, the interested reader is referred to \cite{MEIbook,STRDUI} for more details on modelling of general EMS and to \cite{NAM,ORTbook} for the particular case of electric machines.

We consider multiport, electromechanical systems with magnetic and electric fields consisting of $n_{L}$ magnetic ports, $n_C$ electric ports  and $n_M$ mechanical ports, as defined in \cite{MEIbook,STRDUI}. The port variables of the magnetic and electric ports are voltages and currents, denoted as $(v_L,i_L)\in \rea^{n_L}\times \rea^{n_L}$ and  $(v_C,i_C)\in \rea^{n_C}\times \rea^{n_C}$, respectively. The mechanical port variables are  $(F_E,\dot q)\in \rea^{n_M}\times \rea^{n_M}$, where $F_E$ are the mechanical forces of electrical origin and  $\dot q$ the (rotational or translational) velocities of the movable mechanical elements.

There are external electrical sources, through which electrical energy is supplied to the magnetic and electric elements. To simplify the notation, and without loss of generality, we assume that the electrical subsystem is ``fully actuated''\footnote{{The dimension of the electrical coordinate is equal to the one of the input. The terminology ``full actuation" is borrowed from the literature on mechanical systems.}}, in the sense that there are $n_L$ voltage sources $v_L$, and $n_C$ current sources $i_C$---{see Remark {\bf R6} below}. There are also (electrical and mechanical) dissipation elements---which are assumed linear. The energy stored in the system consists of three components: magnetic energy stored in the inductances, electrical energy stored in the capacitors and mechanical energy stored by the movable part inertia. They are defined by the functions $H_L(\lambda,q)$, where $\lambda \in \rea^{n_L}$ is the vector of flux linkages, $H_C(Q,q)$, where $Q \in \rea^{n_C}$ is the vector of electrical charges and $H_M(q,p)$, where $p \in \rea^{n_M}$ is the generalized momenta, respectively. The systems total energy function is given by
$$
H(Q,\lambda,q,p):= H_C(Q,q)+H_L(\lambda,q)+H_M(q,p).
$$
The constitutive relations of the elements are
$$
\begin{aligned}
v_C & =  \nabla_Q H_C(Q,q)\\
i_L & =  \nabla_\lambda H_L(\lambda,q)\\
F_E & =  - \nabla_q [H_L(\lambda,q)+H_C(Q,q)]\\
\dot q & =  \nabla_p H_M(q,p),
\end{aligned}
$$
where the minus sign in $F_E$ reflects Newton's third law.

The equations of motion of the system can be described in port-Hamiltonian (pH) form \cite{VAN} as
\begali{
    \dot{x} & = \calf \nabla H + g u \\
    y  &=g^\top  \nabla H
\label{pH-General}
}
with the state $x:=\col(Q,\lambda,p,q)$, input and output signals
\beeq{
\lab{xuy}
u:= \begin{bmatrix} i_C \\  v_L \end{bmatrix}, ~~
y:=\begmat{R_C^{-1}  \nabla_Q H_C(Q,q) \\ \nabla_\lambda H_L(\lambda,q)}
}
and the constant, interconnection damping and input matrices
$$
\calf:=
\begin{bmatrix}
 - R_C^{-1} & \mathbf{0} & \mathbf{0} & \mathbf{0}\\
 \mathbf{0} & - R_L & \mathbf{0} & \mathbf{0} \\
 \mathbf{0} & \mathbf{0} & \mathbf{0} & I \\
 \mathbf{0} & \mathbf{0} & -I & - R_M
\end{bmatrix},\; g:=
\begin{bmatrix}
R_C^{-1} & \mathbf{0}\\
\mathbf{0} &   I\\
\mathbf{0}& \mathbf{0} \\
\mathbf{0}& \mathbf{0}
\end{bmatrix},\;
$$
where $R_L$ and $R_C$ are {positive definite}, dissipation matrices, and {$R_M$ is a positive semidefinite mechanical dissipation matrix}.  The class of electromechanical systems described by \eqref{pH-General} is quite large.

Notice that we have adopted as system output the so-called ``natural output" \cite{VAN}. Given this definition of output signal the virtual output \eqref{yv} is given as
\beeq{
\lab{yv0}
y_v=\begmat{R_C^{-1} \nabla^2_Q H_C(Q,q) R_C^{-1}\frak{b}_{CE}  \\ \nabla^2_\lambda H_L(\lambda,q)\frak{b}_{LE}}=:\begmat{y_{vC} \\ y_{vL}},
}
where, for ease of future reference, we introduced $\mathfrak{b}=\col( \frak{b}_{CE} , \frak{b}_{LE})$, with $\frak{b}_{CE} \in \rea^{n_C}$ and $\frak{b}_{LE} \in \rea^{n_L}$ free, scaling vectors and we partitioned $y_v$.

\subsection{Observer for electrical coordinates using $y_v$}
\lab{subsec42}
%
In this subsection, we design an observer for the electrical state $x_E:=\col(Q,\lambda)$ with the help of the virtual output $y_v$, whose estimate is obtained with the filter \eqref{vir_out_filter}.

A key observation in the solution of the problem is that the derivative of $x_E$ is {\em known}. Indeed, from \eqref{pH-General} we have that
\begequ
\lab{dotxe}
\dot x_E =  - R_E y + g_E u,
\endequ
where we defined $R_E:= \diag(R_C^{-1}, R_L)$ and $g_E:= \diag(R_C^{-1} , I) $. {This stems from the fact that this derivative equals the voltage applied to inductors and the currents of capacitors that are the external sources to the system, hence measurable. This fact is a key property to design gradient-based observers \cite{ORTetalcst,ORTYI}, as well as parameter estimation-based observers \cite{ORTetalscl}.} In order to derive our observer we make the additional assumption that the electrical energy terms are quadratic functions of the form
\begali{
H_C(Q,q)  =  \hal Q^\top  C^{-1}(q) Q , \quad
H_L(\lambda,q)  =  \hal \lambda^\top  L^{-1}(q)\lambda,
\lab{enefun}
}
where $L(q) \in \rea^{n_L \times n_L}$ and $C(q) \in \rea^{n_C \times n_C}$ are the positive-definite, inductance and capacitance matrices, respectively.

For this class of energy functions,  the natural output \eqref{xuy} and virtual output \eqref{yv0} take the simpler form
\begali{
y  =\begmat{R_C^{-1} C^{-1}(q)Q  \\  L^{-1}(q)\lambda}, \quad
y_v  =\begmat{R_C^{-1} C^{-1}(q)R_C^{-1}\frak{b}_{CE}  \\  L^{-1}(q)\frak{b}_{LE}},
\lab{yv1}
}
with the partition $y_v:= \col(y_{vC}, y_{vL})$. To streamline the presentation of the observer we define the matrices
\begali{
\nonumber
Y_v &:= \diag( y^\top_{v_C}R_C , y^\top_{v_L}),\; \hat Y_v  := \diag( \hat y^\top_{v_C}R_C , \hat y^\top_{v_L} )\\
\calb_E & := \diag( \frak{b}^\top_{CE} , \frak{b}^\top_{LE} ),
\label{Yvbeta}
}
where $\hat y_v$ is obtained with the virtual output filter \eqref{vir_out_filter}.
\begin{proposition}
\lab{pro4}
\rm
Consider the EMS \eqref{pH-General} with energy function \eqref{enefun}, input signal \eqref{eq1}, bounded state trajectories and the natural and virtual outputs \eqref{yv1}.  Define the observer
\begequ
\label{obs-mutual}
\dot{\hat{x}}_E =  - R_E y + g_E u+ \gamma \hat {Y}^\top_v \big( \calb_E y - \hat {Y}_v \Hat{x}_E \big),
\endequ
where $\hat {Y}_v$ and $\calb_E$ are given in \eqref{Yvbeta} and $\gamma>0$ is a tuning gain. If $Y_v$ is PE, we then have
$$
\lim_{t\to\infty} |\Hat{x}_E(t) - x_E(t)| = \mathcal{O}(\varepsilon).
$$
\hfill{$\triangleleft$}
\end{proposition}

\begin{proof}
\rm
From \eqref{yv1}, \eqref{Yvbeta} and the definition of $x_E$ we get the {\em linear regression} in unknown $x_E$
\beeq{
\lab{regfor1}
\calb_Ey=Y_vx_E,
}
{where we have used
$$
\begin{aligned}
\mathcal{B}_E y  = \begin{bmatrix}\mathfrak{b}^\top_{CE} & 0 \\ 0 & \mathfrak{b}_{LE}^\top\end{bmatrix}
\begin{bmatrix} R_C^{-1} C^{-1}(q) Q \\ L^{-1}(q)\lambda \end{bmatrix}
 = \begin{bmatrix} \mathfrak{b}_{CE}^\top R_C^{-1} C^{-1}(q) Q \\ \mathfrak{b}_{LE}^\top L^{-1}(q)\lambda  \end{bmatrix},
\end{aligned}
$$
and
$$
Y_v x_E = \begmat{y_{vC}^\top R_C & 0 \\ 0 & y_{vL}^\top}  \begmat{Q \\ \lambda}
=
\begmat{(R_C^{-1} C^{-1}(q) R_C^{-1} \mathfrak{b} _{CE})^\top R_C Q\\
(L^{-1}(q) \mathfrak{b}_{CE})^\top \lambda
},
$$
as well as invoking that $R_C$ and $L(q)$ are positive definite.} Substituting \eqref{regfor1} in \eqref{obs-mutual}, and using \eqref{dotxe}, yields
$$
\dot{\tilde{x}}_E = - \gamma \hat {Y}^\top_v \big(Y_v x_E- \hat {Y}_v \Hat{x}_E \big),
$$
where we defined the observation error $\tilde{x}_E:=\hat{x}_E-x_E$. This equation can be written in the form
\begequ
\label{error-eqchi}
\dot{\tilde{x}}_E = - \gamma Y^\top_v Y_v \tilde{x}_E + \xi,
\endequ
where, invoking \eqref{tilyv} and boundedness of all signals, the disturbance term verifies
$
\lim_{t\to\infty}|\xi(t)| \leq \mathcal{O}(\varepsilon).
$
The proof is completed recalling that, under the PE assumption, the unperturbed error equation \eqref{error-eqchi} is exponentially stable \cite{SASBOD} and using standard perturbation arguments.
\qed
\end{proof}
\subsection{Discussion}
\label{subsec43}
%
The following remarks are in order.

\noindent {\bf R6} The class of systems for which the observer of Proposition \ref{pro4} is applicable, can be extended in several directions. First, the assumption of ``fully actuated" electrical coordinates was introduced only to simplify the notation. In the ``underactuated" case, {{\em i.e.}, $\text{rank}(g) < \dim(x_E)$}, two additional, input selecting, tall, constant matrices appear in the definition of the input matrix $g$, they can be ``removed" with a suitable selection of the scaling vector $\frak{b}$, without affecting the results. Second, with some additional calculations it is possible to consider magnetic energy functions of the form
$$
H_L(\lambda,q)  =  \hal [\lambda - \mu(q)]^\top  L^{-1}(q)  [\lambda - \mu(q)],
$$
where $\mu(q)$ represents the flux linkages due to permanent magnets. {For the cases considering magnetic saturation or mutual capacitance/inductance, it is possible to construct a \emph{nonlinear regressor} on $x_E$ instead of \eqref{regfor1}, thus a locally convergent observer would be obtained.}

\noindent {\bf R7} We have worked out the details of an even more general case, namely the EMS shown in Fig. \ref{fig:elemag}. This system is studied in \cite{MEIbook} [Exercise 3-14, pp. 146], where it is assumed that the capacitance and inductance depend, not only on the mechanical position $q \in \rea^2$, but in the capacitor voltage and the inductor current. Hence, the electrical energy functions $H_C(q_1,Q)$ and $H_L(q_2,\lambda)$ are of the form
$$
H_C = \int_{0}^{Q} \dot{\lambda}(Q',q_1) dQ' , \;
H_L= \int_{0}^{\lambda} \dot{Q}(\lambda', q_2) d \lambda'.
$$
with
$$
Q= {c_1 \over c_2 + c_3 q_1}\dot{\lambda}^3, \; \lambda =  (l_1 + l_2 q_2^2)\dot{Q}^3.
$$
Unfortunately, in this case the regression form \eqref{regfor1} is {\em nonlinear} in $x_E$---a case that can still be handled with the proposed method.

\begin{figure}[h]
  \centering
  \includegraphics[width=9cm]{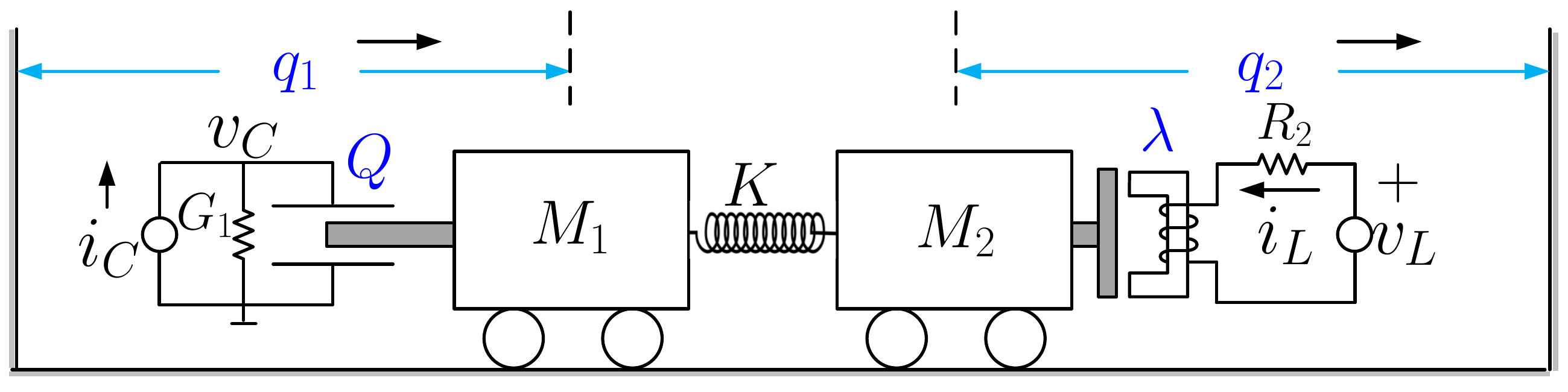}
  \caption{An electromagnetic-field device}\label{fig:elemag}
\end{figure}

\noindent {\bf R8} As shown in the derivations above, the mechanical dynamics plays no role in the solution of the problem of observation of the electrical coordinates. Of course, it is essential to reconstruct $q$ and $p$ from the electrical coordinates---this task is carried out, for completeness, in the examples of the next section. See also  \cite{BOBetal,ORTYI,PYRetal}.

\section{Examples}
\lab{sec5}
%
In this section we apply the result of Proposition \ref{pro4} to two physical examples. For the sake of completeness we also give observers for the mechanical coordinates. 

\subsection{Micro electromechanical optical switch}
\lab{subsec51}
%
The first example is the micro electromechanical optical switch system \cite{BORetal}. This system has only electric-field energy, and its dynamics is described by \eqref{pH-General} with $n_L=0$, $n_C=n_M=1$, that is,
$$
\begmat{\dot{Q} \\  \dot{q} \\ \dot{p}}
 =
 \begmat{- {1\over R_C} & 0 & 0 \\ 0 & 0 & 1 \\ 0 & -1 & -R_M} \nabla H
+ \begmat{{1\over R_C} \\ 0 \\ 0} u,
$$
and the total energy function
$$
H(Q,q,p) = {1\over 2m}p^2 + {a_1\over2}q^2 + {a_2\over4}q^4 +{1\over 2C(q)}Q^2,
$$
where $m>0$ is the mass of the actuator, $a_1,a_2>0$ are the spring constants, $C(q)={c_1(q+c_0)}$ and $c_0,c_1>0$ are the capacitance parameters. A physical constraint is $q>0$. The output is the voltage in the capacitor
$
y=v_C = {1\over c_1(q+c_0)}Q,
$
and the virtual output, with $\frak{b}=1$, is $y_v ={1\over R_Cc_1 (q+c_0)}$. Note that from the virtual output we can directly recover $q$. The linear regressor \eqref{regfor1} takes, then, the form
\beeq{
\lab{regfor2}
{1\over R_C} y = y_v Q,
}
with the state $Q$ to be estimated. A full state observer design is given below, and we assume that a simple projection operator of $\hat{y}_v$ has been adopted, but omitted for brevity, to guarantee $\hat{y}_v \neq 0$.
\begin{proposition}
\label{pro5}\em
%
For the micro electromechanical optical switch model, the observer
\begequ
\label{obs9}
\begin{aligned}
    \dot{\Hat{Q}} & = - {1\over R_C}y + {1\over R_C}u + \gamma \hat{y}_v \bigg( {1\over R_C} y - \hat{y}_v \Hat{Q}\bigg) \\
    \Hat{q} &= {1\over R_Cc_1\cdot \hat{y}_v} - c_0\\
    \dot{\Hat{p}} & = -a_1\hat{q} - a_2\hat{q}^3 + {1 \over 2c_1(\hat{q}+c_0)^2} \hat{Q}^2 - {R_M\over m}\hat{p}.
\end{aligned}
\endequ
with $\gamma >0$ guarantees
$$
\lim_{t\to\infty} \left|
\col(
\tilde{Q} , \tilde {q}, \tilde{p})
\right| =\mathcal{O}(\varepsilon), \; \text{(exp.)}
$$
with the observation errors
$
\tilde{Q}:= \Hat{Q} - Q,\; \tilde{q}: = \Hat{q} - q,\;\tilde{p}: = \Hat{p} - p.
$
\hfill{$\triangleleft$}
\end{proposition}
\begin{proof}
\rm
Replacing \eqref{regfor2} in the first equation of \eqref{obs9} yields the error equation
$$
\dot{\tilde{Q}} = - \gamma y_v^2 \tilde{Q} + \mathcal{O}(\varepsilon).
$$
Since $q>0$, we have that $y_v>0$, ensuring the exponential convergence of $\tilde{Q}$ to a small neighborhood of the origin. Also, notice that $\lim_{t\to\infty}\hat q(t) - q(t) = \mathcal{O}(\varepsilon)$.

Now, given the fact that
$$
\dot{{p}}  = -a_1{q} - a_2{q}^3 + {1 \over 2c_1({q}+c_0)^2} {Q}^2 - {R_M\over m}{p},
$$
the third equation in \eqref{obs9} may be written as
$$
\dot{\tilde {p}}  =- {R_M\over m}\tilde p + \epsilon_t + \mathcal{O}(\varepsilon).
$$
This completes the proof.
\qed
\end{proof}

\subsection{Levitated ball: Simulation and experimental results}
\lab{subsec52}
%
The dynamics of classical 1-dof MagLev ball is described by the pH system \eqref{pH-General} with $n_C=0$, $n_L=n_M=1$, and the Hamiltonian
$$
H(\lambda,q,p) = mGq + {p^2 \over 2m} + {\lambda^2 \over 2L(q)},
$$
where $m >0$ is the mass of the ball, $G$ is the gravity constant and the inductor is assumed of the form
$
L(q) = {k \over c- q },
$
with $c>0$ and $q\in (-\infty, c)$. The full dynamics is given by
$$
\begmat{\dot{\lambda} \\ \dot{q} \\ \dot{p}}
=
\begmat{ - R  & 0 & 0 \\ 0 & 0 & 1 \\ 0 & -1 & 0} \nabla H
+
\begmat{1 \\ 0 \\ 0 } u.
$$

The output signal is the current in the inductance, that is,
\beeq{
\lab{outmaglev}
y=i_L={\lambda \over k}(c-q),
}
and the virtual output is $y_v={c-q \over k}$. Notice that, similarly to the previous example, from the virtual output we can directly recover the ball position $q$.
\subsubsection{Adaptive observer design}
\lab{subsubsec521}
%
In this subsection, we address the more challenging problem of \emph{adaptive} state observation for this system, namely, with the parameters $m$, $c$ and $k$ known, but $R$ {\em unknown}.

Before presenting the resistance estimator we recall the physical constraints $q \in (-\infty,c)$. As expected, we impose this constraint also to its estimate,\footnote{This can easily be done adding a projection operator to the second equation in \eqref{vir_out_filter}, but is omitted for brevity.} hence
\begali{
y_v =q - c < 0 , \quad
\hat y_v < 0.
\lab{yvneqzer}
}
To simplify the notation in the sequel we introduce a change of coordinate for the position and, denote
$
x:=  \col(\lambda,{1 \over k}(c-q),p),
$
for which we have the following adaptive observer.

\begin{proposition}
\label{prop_MagLev}
\em
Consider the $1$-dof MagLev system. Assume $i_L$ is PE. Then, the adaptive observer
\begequ
\label{obs_MagLev}
\begin{aligned}
\dot{\Hat{R}} & = \gamma_R \phi_R \big(Y_R - \phi_R \Hat{R}\big)
\\
\dot{\hat{x}}_1 & = {-{ \hat{R}} y} + u - \gamma_\lambda  (y - \hat{y}_v \hat x_1)
\\
  \dot{z} & =-{{ \gamma_p  \over km}} z + {1 \over 2k} \hat x_1^2 + { {\gamma_p^2  \over km}} y_v - mG\\
\Hat{x}_2 & = \hat{y_v}\\
    \Hat{x}_3 & = z {- \gamma_p  \hat{y}_v},
\end{aligned}
\endequ
where
\begali{
\nonumber
\dot{v}_1 & = - a v_1 + au \\
\dot{v}_2 &  = - a v_2 + a{y\over \hat{y}_v}  \nonumber \\
\dot \phi_R & = -a \phi_R + {a}y \nonumber \\
Y_R & = -v_1 + a{y\over \hat{y}_v} - av_2,
\lab{stareafil}
}
with $a, \gamma_R , \gamma_\lambda, \gamma_p >0$, guarantees
\begalis{
\lim_{t\to\infty} \big| \tilde{R}(t)  \big|  & = \mathcal{O}(\varepsilon),\;\lim_{t\to\infty} \big| \tilde{x}(t)   \big|= \mathcal{O}(\varepsilon),\; \text{(exp.)},
}
where we defined the parameter estimation and state observation errors $\tilde R:=\hat R-R$ and $\tilde x:= \hat{x} - x$, respectively.
\hfill{$\triangleleft$}
\end{proposition}

\begin{proof}\rm \footnote{We give the sketch of the proof in the nominal case $\hat{y}_v= {y}_v$. The full proof can be obtained via perturbation analysis, which can be found in \cite{YIetalcdc}.}
 From \eqref{outmaglev} and \eqref{yv} we have that $\lambda = {y \over y_v}$, which is well defined in view of \eqref{yvneqzer}. Computing the derivative with respect to time yields 
 $$
 {d \over dt}\left( {y \over y_v} \right) = - R{y}+u .
 $$ 
 Applying to the equation above the LTI filter ${a \over \bp+ a}$ yields
\begequ
\lab{linregr}
 {a \bp \over \bp+ a}\bigg[ {y \over y_v} \bigg] - { a \over \bp+ a} \big[u\big] = - R { {a} \over \bp+ a} \big[y\big]  + \et,
\endequ
As shown in \cite{ARAetaltac}, without loss of generality, the term $\et$ is neglected in the sequel.

Notice now that \eqref{stareafil} is a state realization of the filters
\begequ
\lab{yrphir}
\begin{aligned}
    Y_R  = {a \bp \over \bp+ a}\bigg[ {y \over y_v} \bigg] - { a \over \bp+ a} \big[u\big] , \quad
    \phi_R  =  {-  {{a}  \over \bp+ a} \big[y\big].}
\end{aligned}
\endequ
From \eqref{linregr} and \eqref{yrphir} we get the linear regression
$
Y_R=R \phi_R,
$
which upon replacement in the first equation of \eqref{obs_MagLev}, yields
$$
\dot{\Tilde{R}} = - \gamma_R \phi_R^2 \Tilde{R},
$$
where we defined the resistance estimation error $\Tilde{R}:= \Hat{R} - R$. Exponential convergence {to zero} follows invoking the PE assumption of $y$, which ensures $\phi_R$ is also PE \cite{SASBOD}.

We proceed now to analyze the behavior of the state observer. From the second equation of \eqref{obs_MagLev}, we get
\begalis{
    \dot{\tilde {x}}_1  = - \gamma_\lambda  (y - y_v \hat x_1)+ { \tilde{R}} y
     = - \gamma_\lambda  y^2_v \tilde x_1+ { \tilde{R}} y.
}
From \eqref{yvneqzer} we see that the unperturbed dynamics is exponentially stable and, moreover, $y$ is bounded. Using these two properties and the fact that $\lim_{t\to\infty}\tilde R(t) = 0$ (exp.) proves that $\lim_{t\to\infty}\tilde x_1(t) = 0$ (exp.).

The proof is completed noting, after some lengthy but straightforward calculations, that for $x_3$ we get the error equation
$$
\dot {\tilde x}_3  =- {{\gamma_p   \over km}}\tilde x_3 + {1 \over 2k}(\hat x_1 + x_1)\tilde x_1,
$$
and recalling that the term in parenthesis in the right hand side is bounded.
\qed
\end{proof}
\subsubsection{Discussion}
\lab{subsubsec522}
%
The following remarks are in order.

\noindent {\bf R9}  The electromagnetic valve actuator is another magnetic-field EMS extensively used in industry \cite{PETSTE} for which the observer of Proposition \ref{prop_MagLev} can be directly applied.

\noindent {\bf R10} We make the important observation that it is possible to show that the MagLev system, with the output $y$ given in \eqref{outmaglev}, does not satisfy the observability rank condition [Section 1.2.1]\cite{BESbook}, therefore it is not uniformly differentially observable.

\noindent {\bf R11} The assumption that $i_L$ is PE is not restrictive at all. Actually it is possible to show that this condition can be transferred to the control $u$.\footnote{The details of this proof are omitted for brevity.} Now roughly speaking, since $u$ defined in \eqref{eq1} contains an additive {(high-frequency)} term that is PE, the condition that $u$ is PE will almost always be satisfied.

\noindent {\bf R12} The observer presented above is a Kasantzis-Kravaris-Luenberger (KKL) observer, see \cite{ANDPRA}. An alternative to this is a standard Luenberger observer
\begalis{
    \dot z_1 & ={- {1 \over km}} z_2 + c_1({x}_2 - z_1) \\
    \dot z_2 & = {1 \over 2k} \Hat{x}_1^2 - mG + c_2({x}_2 - z_2)\\
    \hat x_3 &= z_2,
}
with $c_1>0$ and $c_2>0$. However, the order of such a design is higher than that of the KKL observer. Moreover, as shown in Subsection \ref{subsubsec523} it was observed in simulations that the KKL observer outperforms the Luenberger one and is easier to tune.
\subsubsection{Simulations}
\label{subsubsec523}
In this subsection, the performance of the observer \eqref{obs_MagLev} for the MagLev system, together with the estimator for virtual output \eqref{vir_out_filter}, are validated via computer simulations conducted with Matlab/Simulink. The parameters used in the simulation are given in Table \ref{tab:2}. The new design is compared, via simulations, with the ones in \cite{COMetalacc,YIetalscl}. In both simulations and experiments, the desired equilibrium is $(\sqrt{2kmg}, q_{\star}, 0)$, with $q_{\star}$ taken as a pulse train, and with the initial states $(\sqrt{2kmg}, 0, 0)$. For a fair comparison with other designs, parameters are tuned to guarantee that the performance at the transient stage are similar. Simulations are run with the {\em full state-feedback} version of the interconnection and damping assignment passivity-based control (IDA-PBC) \cite{ORTetalcsm}, that is:
\begali{
u_C & = - \hat{R} i - K_p \bigg( {1\over \alpha} (\lambda -\lambda_{\star})  + (q - q_{\star}) \bigg)  - \bigg( {\alpha \over m} + K_p \bigg) p,
\nonumber
}
with $K_p = 200.7$ and $\alpha = 33.4$. To make simulations more realistic, we add measurement noise in the current $i$, which is generated with the ``band-limited white noise'' block with the noise power $10^{-10}$ and sample time $10^{-3}$ s. The parameters in the proposed observer are selected as ${\mathfrak{b} }=1,\varepsilon=1/300, a=500, \gamma=3.5\times 10^8,  \gamma_R=500, \gamma_\lambda=8000, \gamma_p= 30, {\hat{R}(0)=2}$, {and the other initial states were selected as zero}. The parameters of the design in \cite{YIetalscl} are selected as $n=10,\varepsilon=1/300, \alpha=0.01,\gamma=50$, and the parameters in \cite{COMetalacc} are selected as $n=10,\varepsilon=1/300$. {The excitation periodic function is selected as sinusoidal function.}
\begin{table}[h]
\centering
\caption{Parameters of MagLev systems: Simulation (First Column) and Experiments (Second Column)}
\label{tab:2}
\renewcommand\arraystretch{1.3}
\begin{tabular}{l|r|r}
\hline\hline
 Ball mass [kg]&  0.0844&  0.0844  \\
 Gravitational acceleration [$\text{m/s}^2$] &9.81& 9.81\\
 Resistance [$\Omega$] &2.52 & 10.615 \\
 Position ($c$) [m] & 0.005& 0.0079\\
 Inductance constant ($k$) [$\mu$H$\cdot$m]& 6404.2 & 49950\\
 \hline\hline
\end{tabular}
\end{table}

Simulation results in Matlab/Simulink are shown in Figs. \ref{fig:virtual}-\ref{fig:state}, where $\hat{y}_v,\hat{\lambda},\hat{q},\hat{p}$ and $\hat{R}$ denote the results from the proposed design, $\hat{y}_v^1,\hat{\lambda}^1$ and $\hat{q}^1$ denote the one from the design in \cite{YIetalscl}, and $\hat{y}_v^2$ is the virtual output estimation from \cite{COMetalacc}. As expected, the new design is less sensitive to measurement noise due to its structure, and also, the steady-state observation are of $\mathcal{O}(\varepsilon)$ accuracy. Besides,  the KKL observer outperforms the Luenberger one, whose estimate is written as $\hat{p}^\texttt{L}$. We then test a {\em sensorless} version of the IDA-PBC law with the same parameters as those above. We observe in Fig. \ref{fig:state-outputfeedback} that the position has a significant regulation error in the first second, which is due to the initial inaccurate estimation of $R$. However, the remaining transients are satisfactory.

\begin{figure}[]
    \centering
\includegraphics[width=7cm]{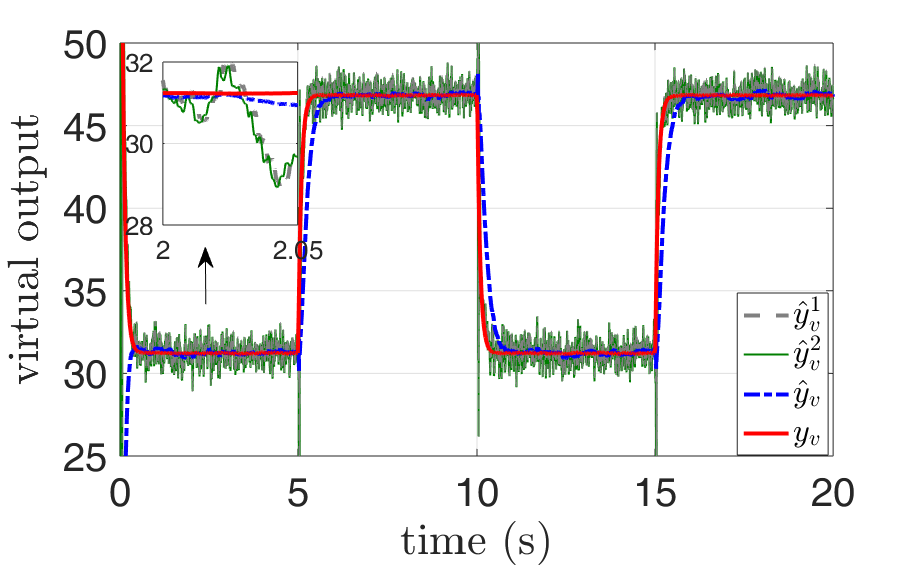}
\includegraphics[width=7cm]{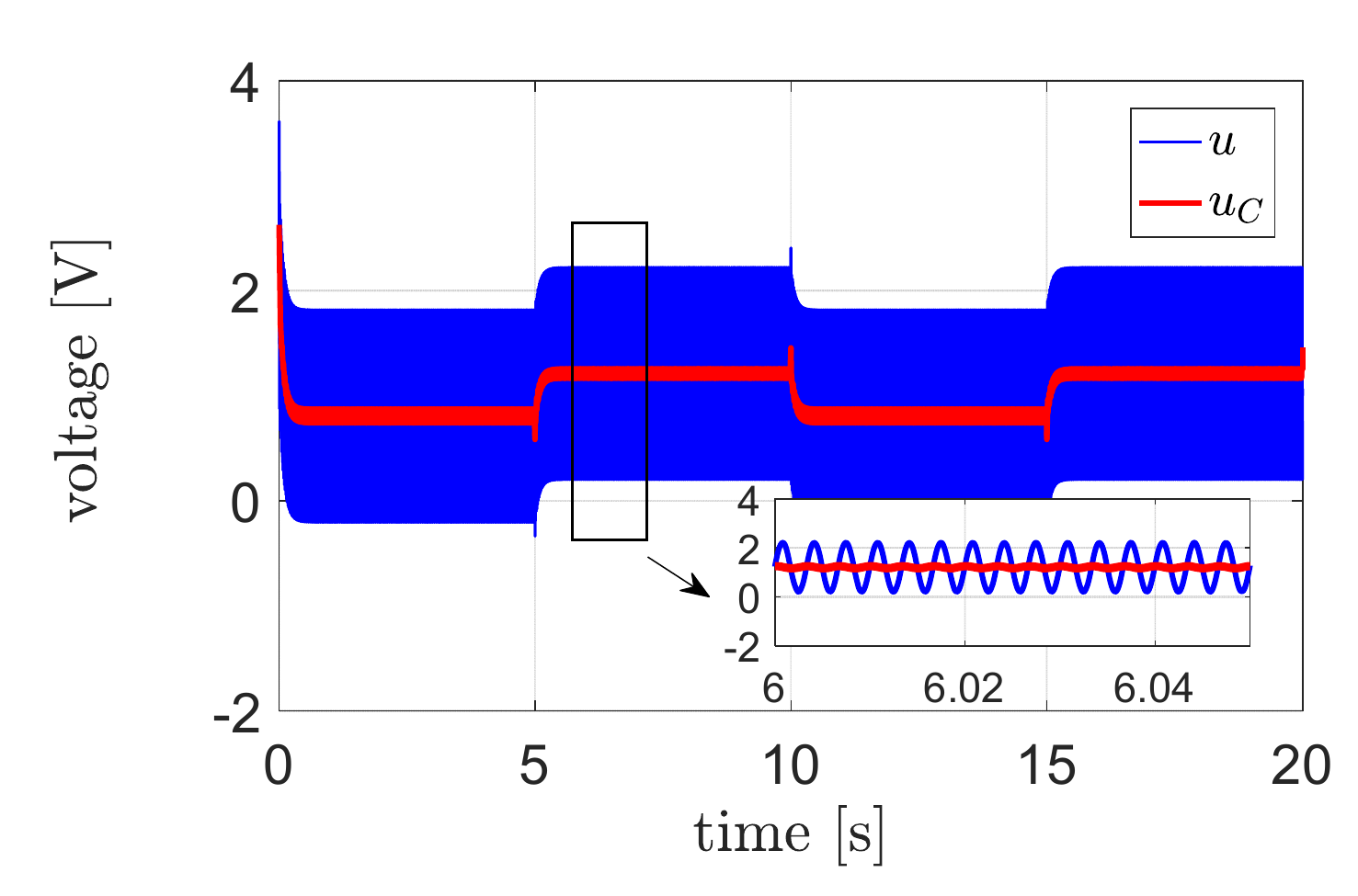}
    \caption{Virtual output estimation and input with signal injections (simulation)}
    \label{fig:virtual}
\end{figure}

\begin{figure}[]
    \centering
\includegraphics[width=7cm]{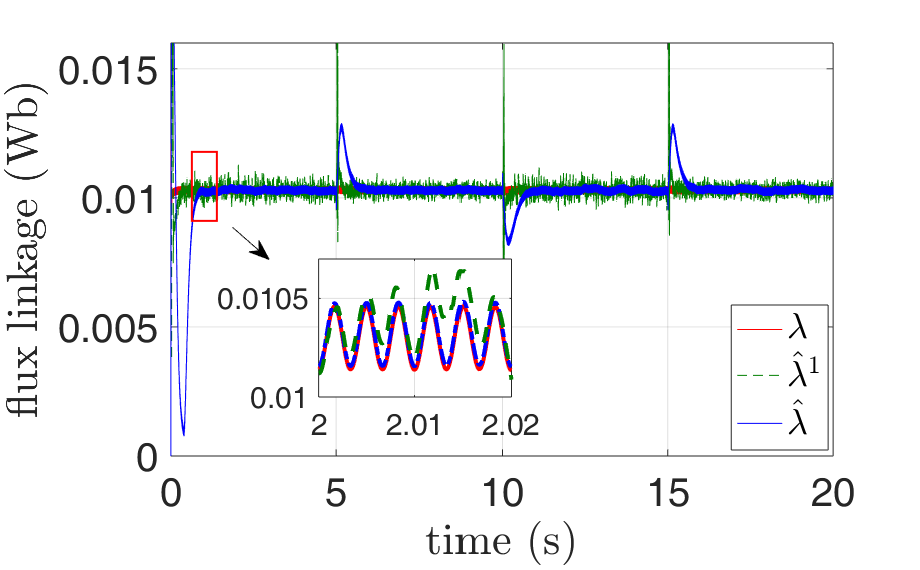}
\includegraphics[width=7cm]{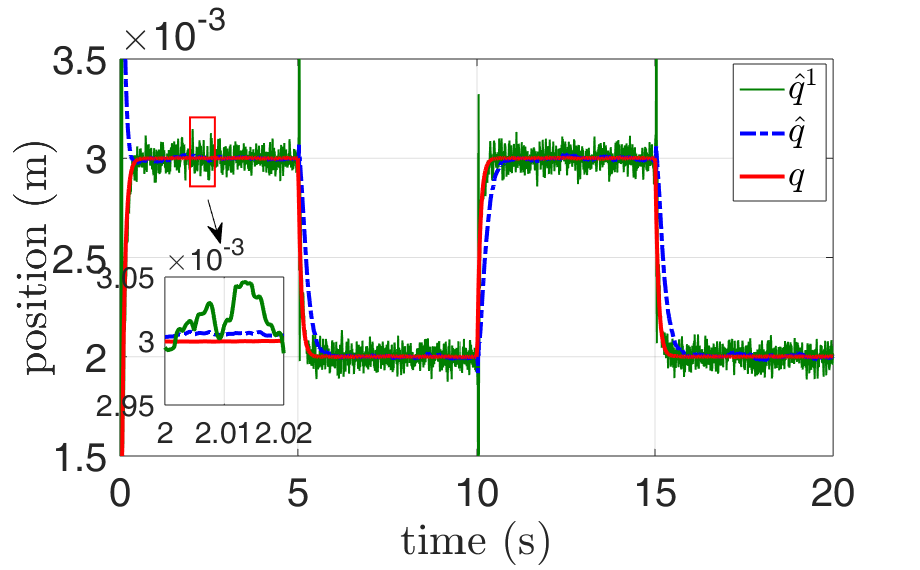}
\includegraphics[width=7cm]{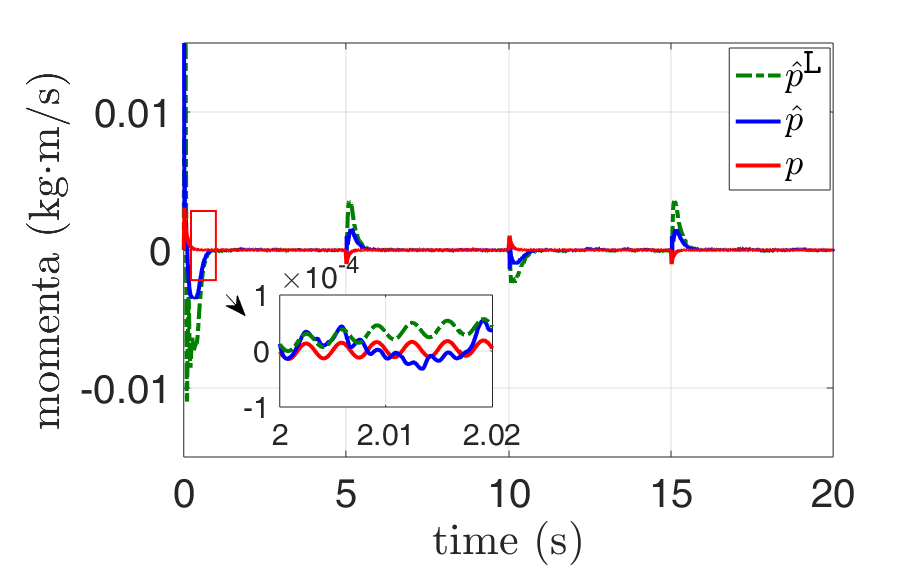}
\includegraphics[width=7cm]{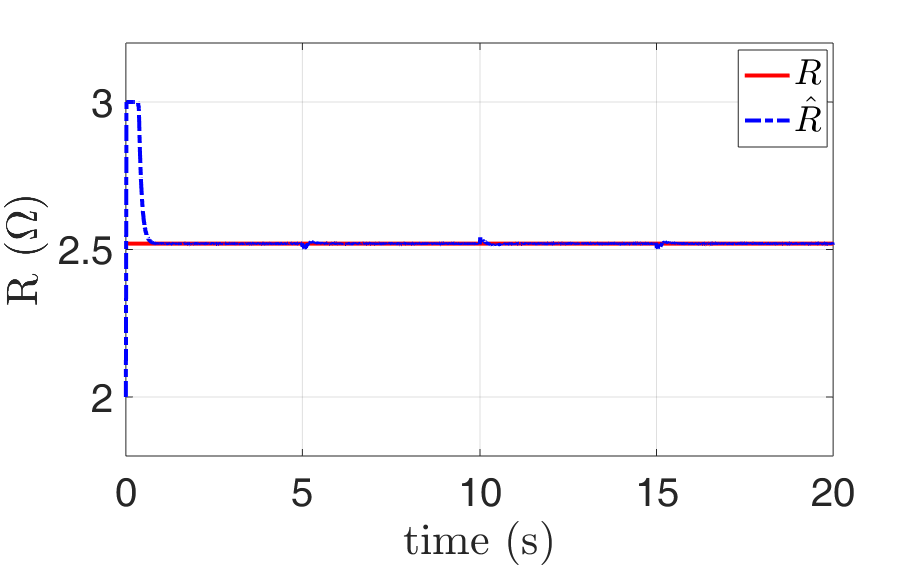}
    \caption{State and parameter estimations (simulation)}
    \label{fig:state}
\end{figure}

\begin{figure}[]
    \centering
\includegraphics[width=6.5cm]{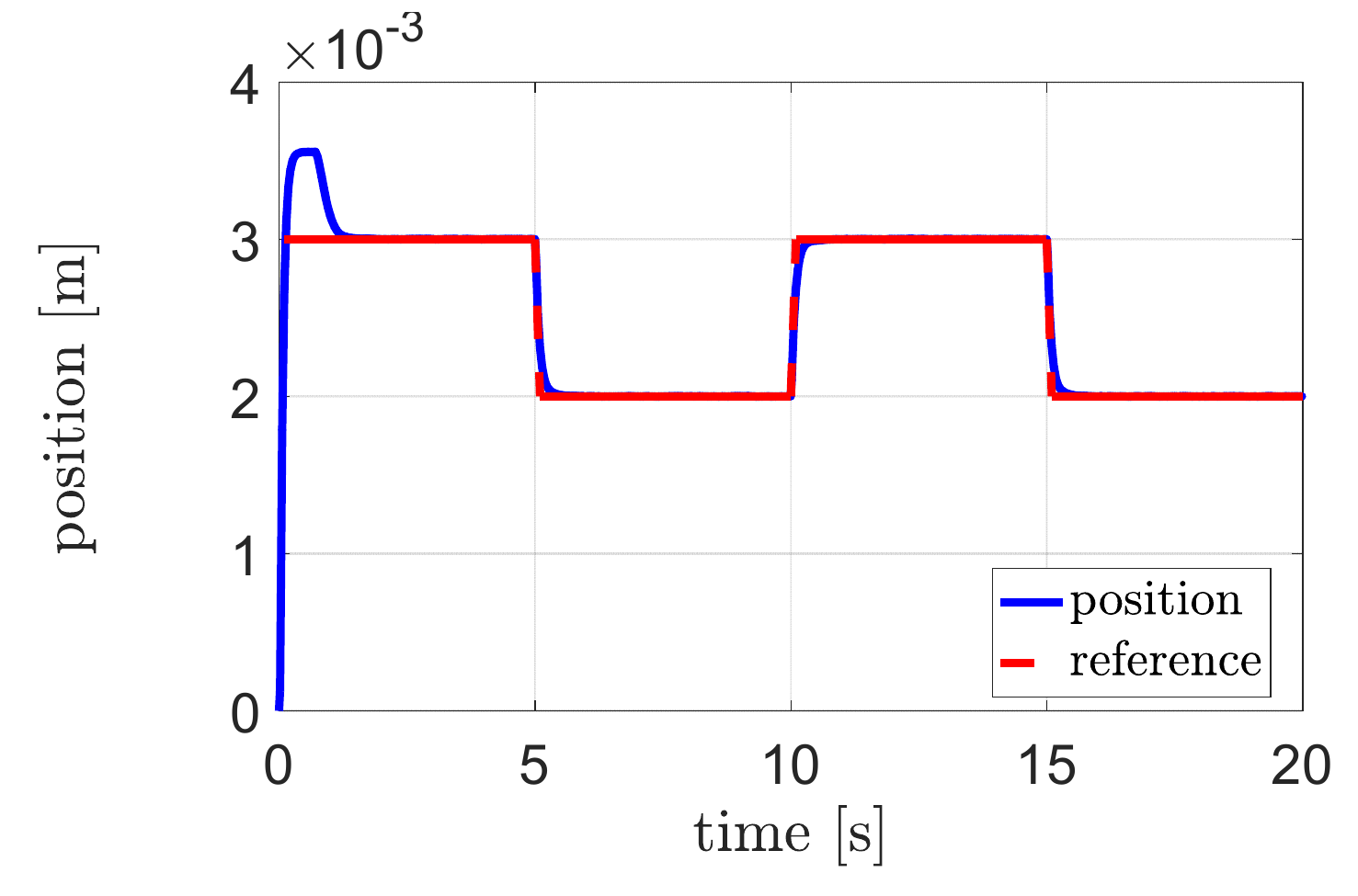}
\includegraphics[width=6.5cm]{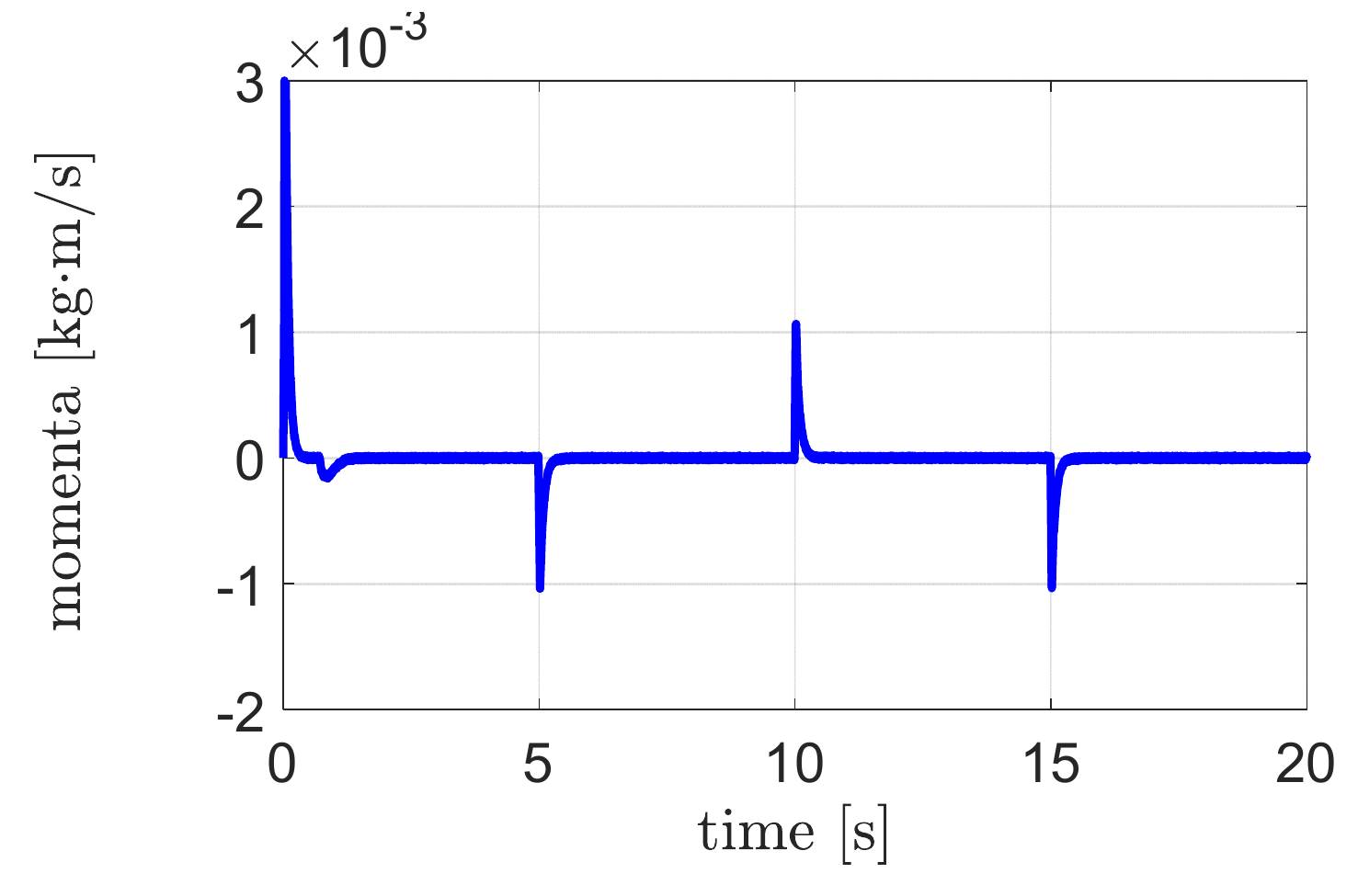}
    \caption{States with sensorless control (simulation)}
    \label{fig:state-outputfeedback}
\end{figure}
\subsubsection{Experiments}
\label{subsubsec524}
Some experiments have been conducted on the experimental set-up shown in Fig. \ref{fig:set-up}, which is located at the D\'epartment d'Automatique, CentraleSup\'elec. The proposed adaptive observer was tested in closed-loop with the well-tuned backstepping-plus-integral controller
$$
\begin{aligned}
u_0 & = R(c-q)|\Upsilon|^{1\over 2}\text{sign}(\Upsilon)  -  K_i u_\texttt{I}\\
\dot{u}_\texttt{I} &= q - q_{\star}, \quad u_\texttt{I}(0) = 0,
\end{aligned}
$$
with $\Upsilon(q,p) = {2\over k} \big( mg - \gamma_1(p-p_{\star} ) - \gamma_2 m(q- q_{\star})\big)$, $K_i=1,\;\gamma_1=340$ and $\gamma_2=3$. The parameters in the observer are taken as $A_0=1.5, \varepsilon=1/33, d=10\varepsilon,a=10$ and $\gamma=4\times 10^5, \gamma_R=50, \gamma_\lambda=8000, \gamma_p=20$.

The responses are shown in Figs. \ref{fig:state_exp}-\ref{fig:output-exp}, where we also give the position estimate $\hat{q}^*$ from the design in \cite{YIetalscl}. Unfortunately, the device is only equipped with sensors for position and current. Hence, we can only compare the position estimate with its measured values, as well as the flux linkage estimate with its desired equilibrium. Again, we verify the accuracy and the robustness of the new observer in the presence of measurement noise. Fig. \ref{fig:freq} gives the position estimates with different probing frequencies. It illustrates that a higher frequency yields a higher accuracy, but at the price of a more jittery response.

\begin{figure}[]
    \centering
\includegraphics[width=8cm]{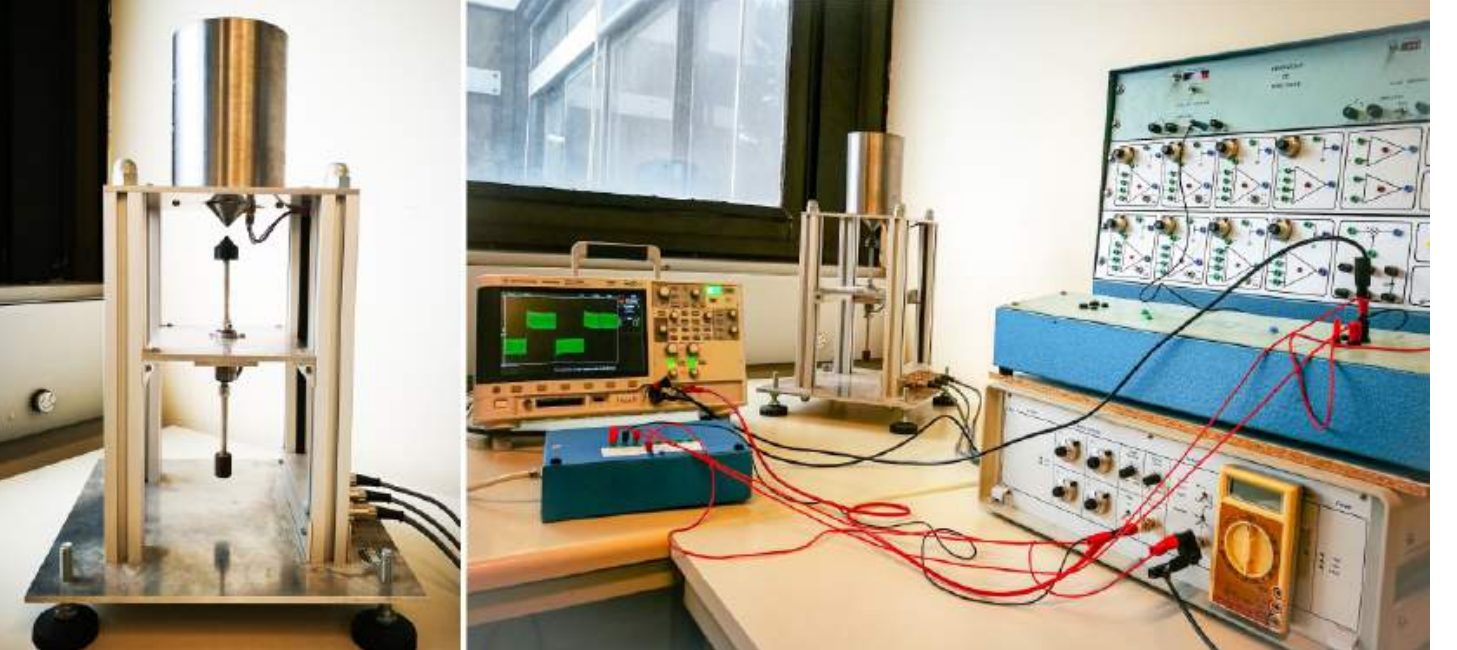}
    \caption{Experimental set-up}
    \label{fig:set-up}
\end{figure}

\begin{figure}[]
    \centering
\includegraphics[width=7cm]{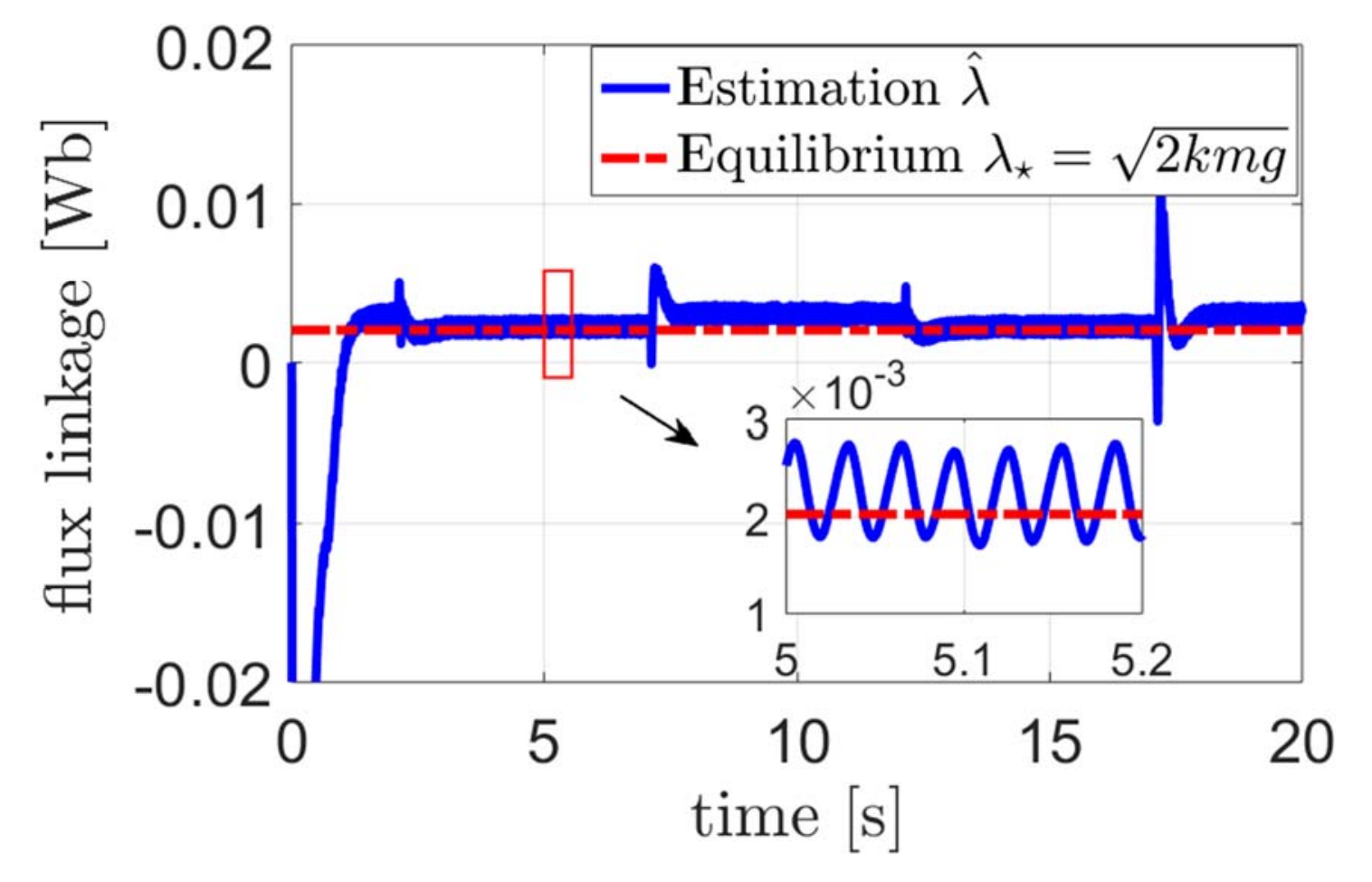}
\includegraphics[width=7cm]{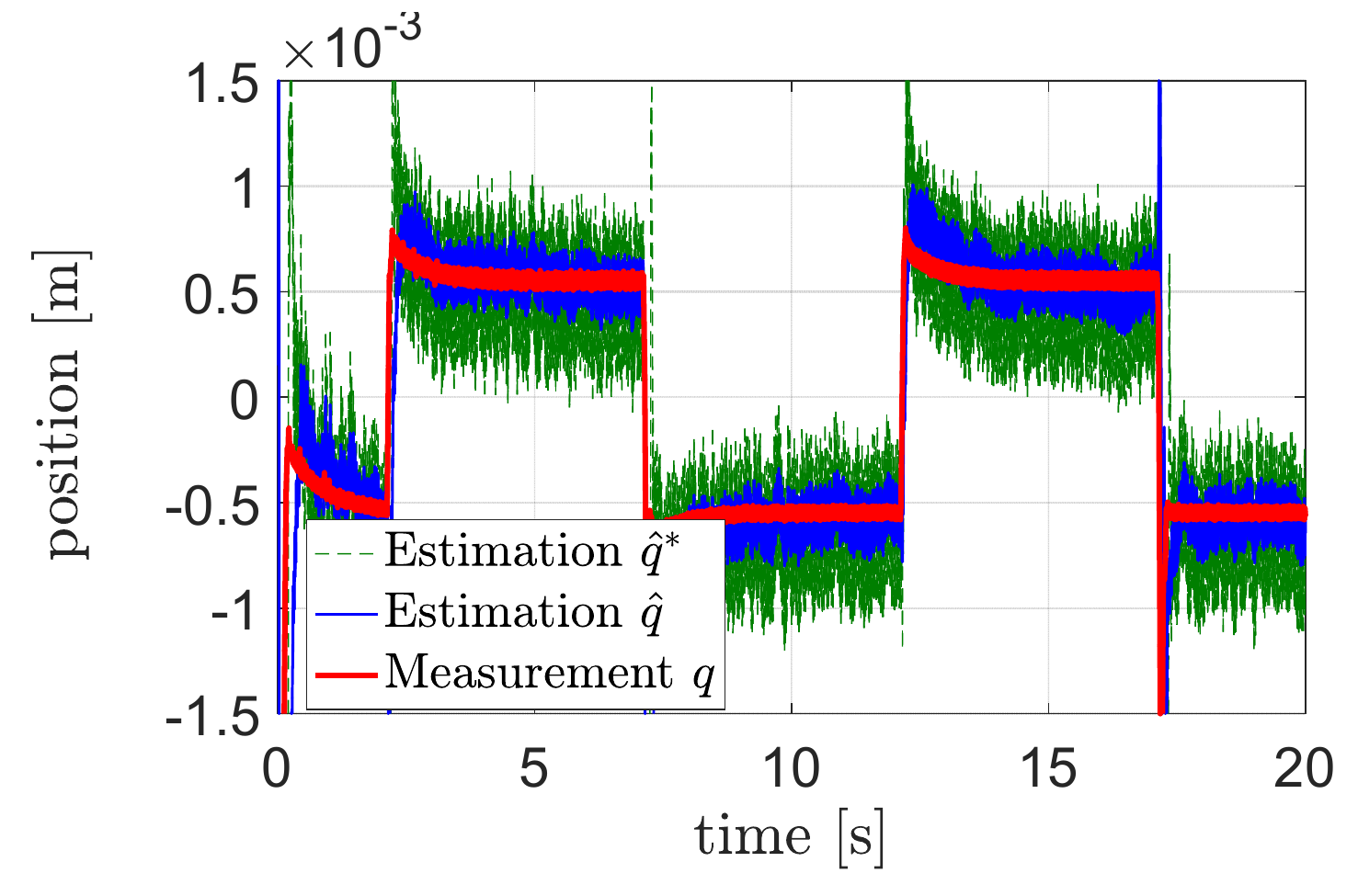}
\includegraphics[width=7cm]{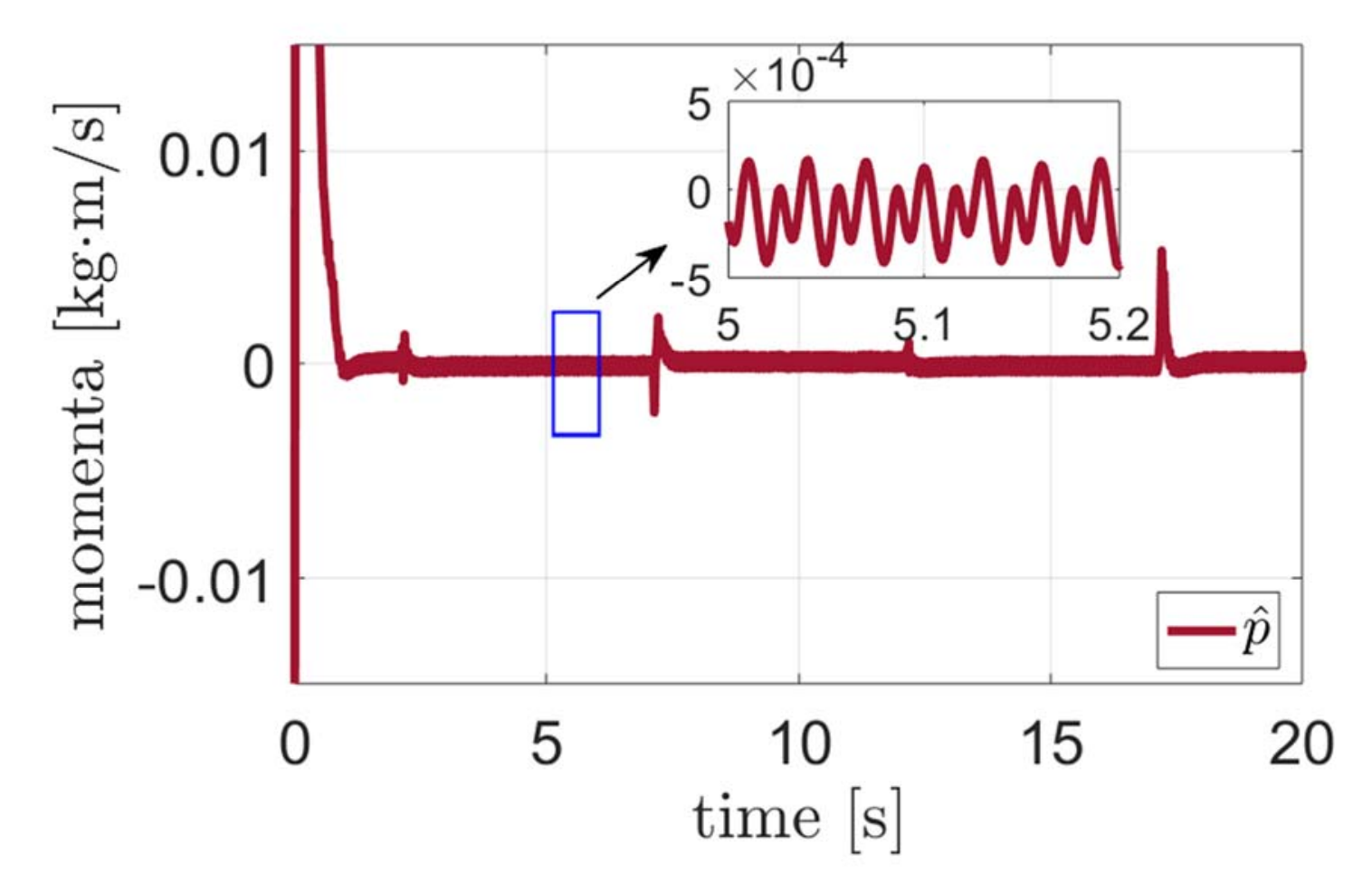}
\includegraphics[width=7cm]{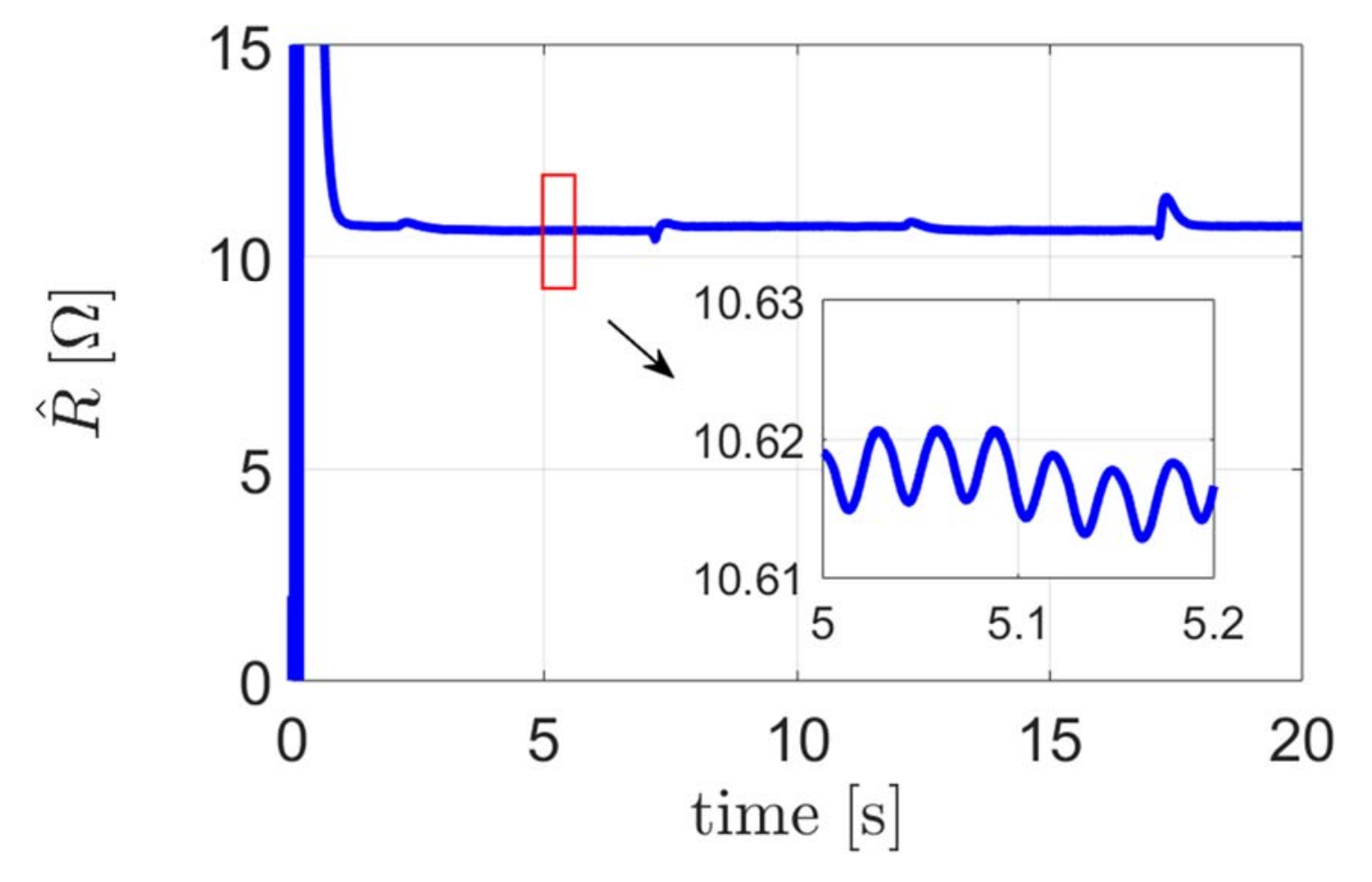}
    \caption{{State and parameter estimations (experiment)}}
    \label{fig:state_exp}
\end{figure}

\begin{figure}[]
    \centering
\includegraphics[width=7cm]{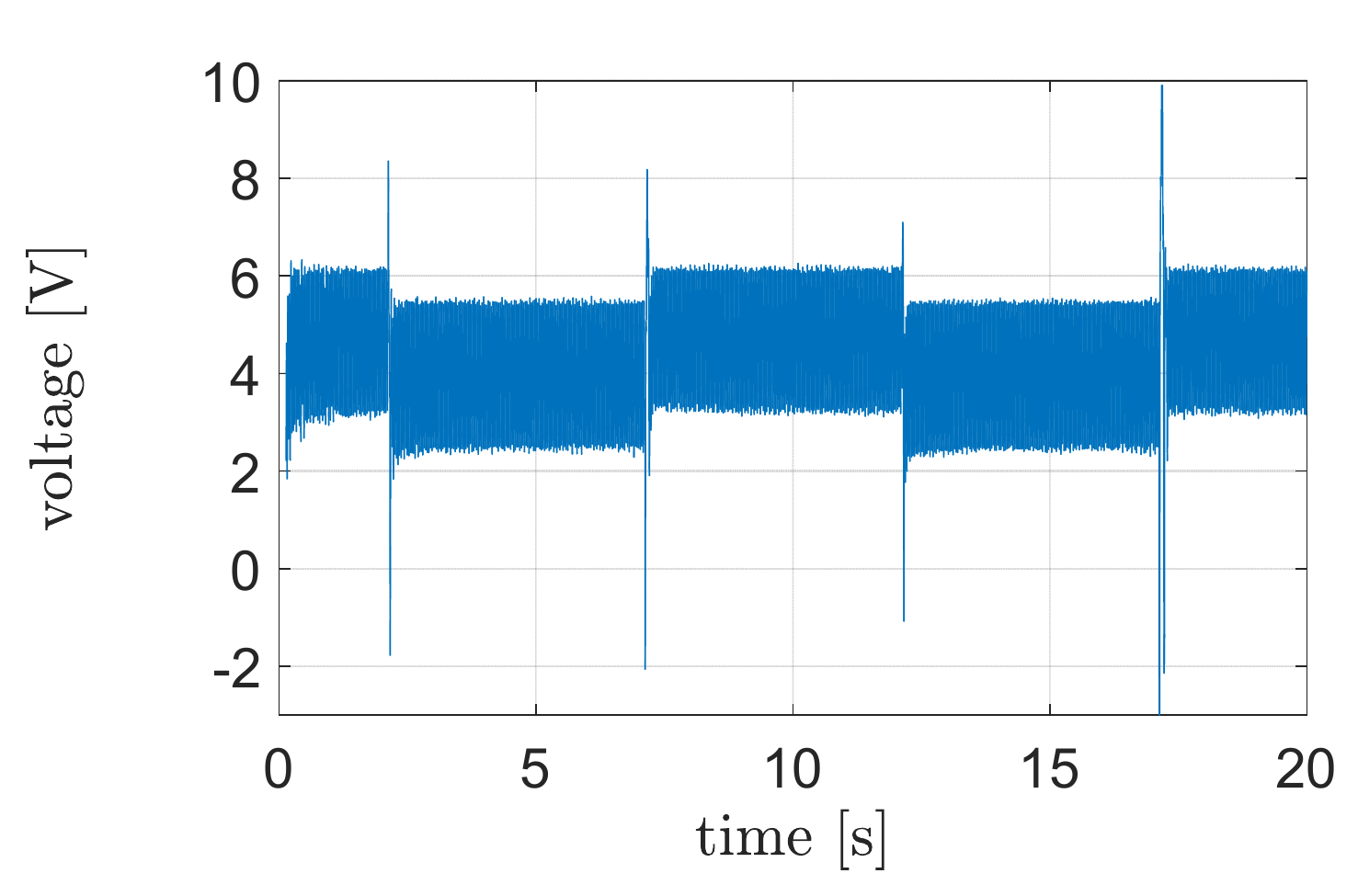}
\includegraphics[width=7cm]{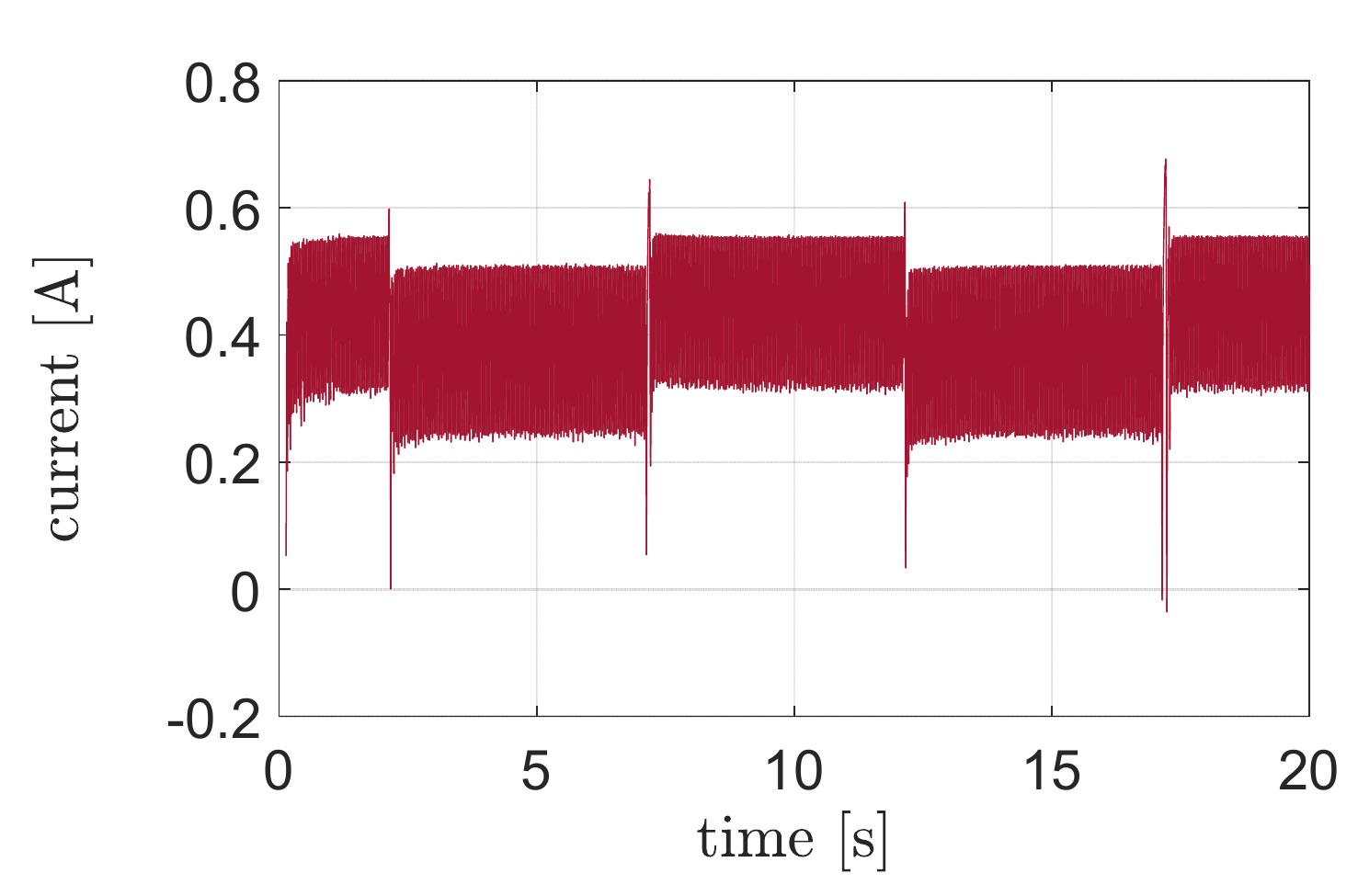}
    \caption{Input and output (experiment)}
    \label{fig:output-exp}
\end{figure}

\begin{figure}[]
    \centering
\includegraphics[width=7cm]{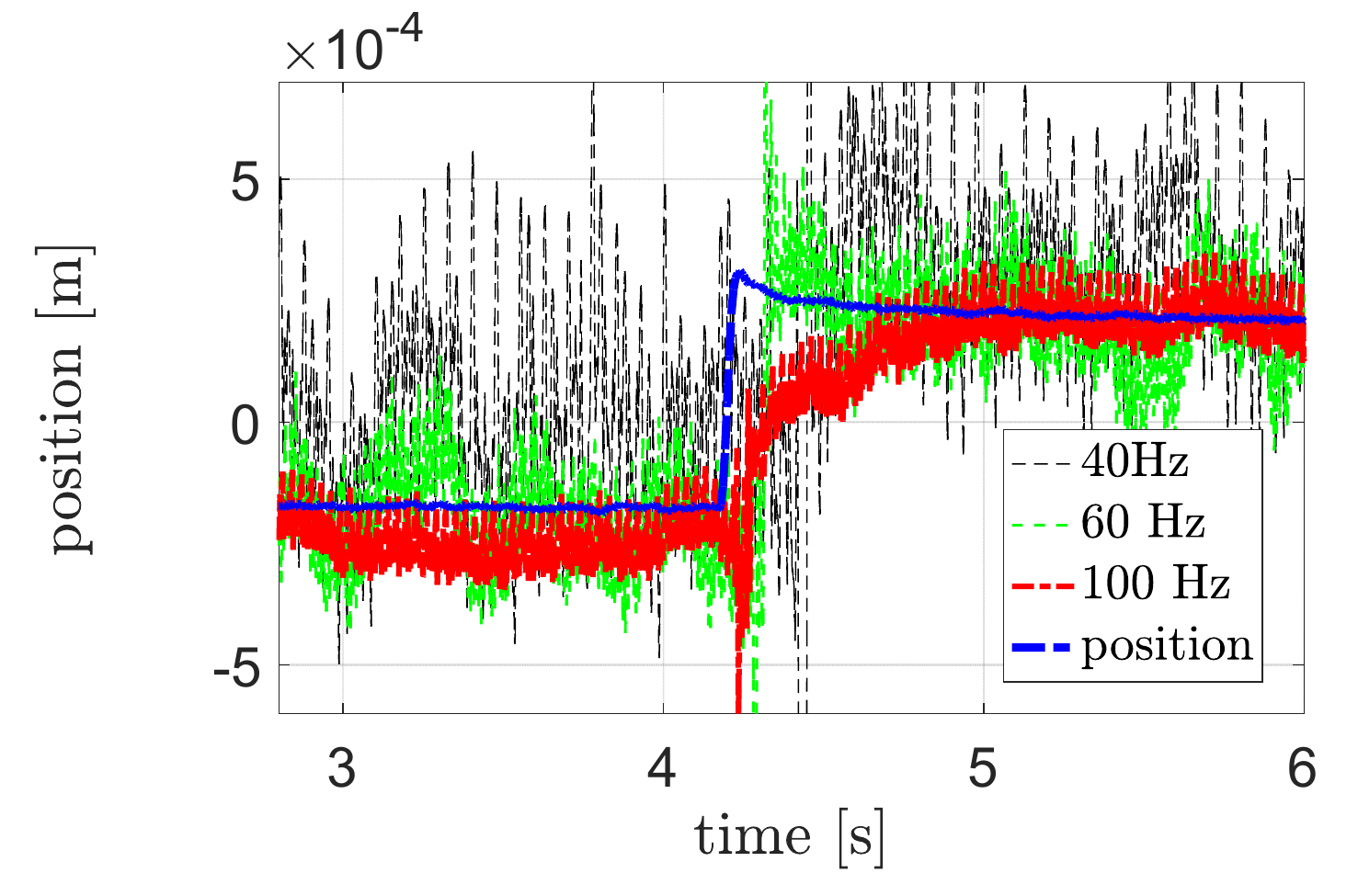}
    \caption{Position estimation for different excitation frequencies (experiment)}
    \label{fig:freq}
\end{figure}

%
\section{Conclusions}
\lab{sec6}
%

In this work we present a new filter for the estimation of the virtual outputs generated with the signal injection technique of \cite{COMetalacc}. The proposed filter has a \emph{closed-loop} structure, providing some robustness to measurement noise and parameter uncertainty. We apply the new design to the sensorless observation problem of a class of EMS. The method is illustrated with two examples---the optical switch and the 1-dof MagLev system, with experiments being conducted for the latter.

Some problems that are being currently investigated are the following.
\begin{itemize}
  \item In \cite{BOBetal,PYRetal} observers for EMS, that do not rely on signal injection, have been proposed. It would be interesting to compare the performance of both approaches and, eventually, combine them in an effective way. In particular, using the signal injection when the signal excitation level is low. A mixed scheme like this has recently been proposed for motors in \cite{CHOetal,ORTYI}.
  \item The 1-dof MagLev system is a benchmark of electromechanical systems. We are currently investigating the application of the new observer to other electromechanical systems---in particular, electrical motors.
  \item Although we analyze the effects of probing frequencies from the theoretical viewpoint, as pointed out in Remark {\bf R2} there are many practical considerations that must be taken into account before claiming it to be an operational technique.
  %
\end{itemize}

%
\section*{Acknowledgement}                            
This paper is supported by the NSF of China (61473183, U1509211, 61627810), China Scholarship Council, National Key R\&D Program of China (SQ2017YFGH001005), and by the Government of the Russian Federation (074U01), the Ministry of Education and Science of Russian Federation  (GOSZADANIE 2.8878.2017/8.9, grant 08-08).

\end{document}